\documentclass[12pt]{article}

\usepackage{amsmath,amssymb}
\usepackage{amsfonts,color}
\usepackage{epsfig}

\newtheorem{theorem}{Theorem}

\newtheorem{lemma}[theorem]{Lemma}

\newtheorem{problem}[theorem]{Problem}

\newenvironment{proof}[1][Proof]{\noindent\textbf{#1.} }{\ \rule{0.4em}{0.4em}}

\renewcommand{\Re}{{\text{Re\,}}}
\renewcommand{\Im}{{\text{Im\,}}}

\newcommand{\grad}{\nabla}
\newcommand{\gradx}{\nabla\hspace{-2.5pt}_x}
\newcommand{\grady}{\nabla\hspace{-2.5pt}_y}

\newcommand{\curl}{\nabla\hspace{-0.25em}\times\hspace{-0.2em}}

\newcommand{\gradp}{{\grad^\perp}}

\newcommand{\gradpy}{{\grad^\perp_{\hspace{-1.5pt}y}\hspace{1pt}}}
\newcommand{\gradpx}{{\grad^\perp_{\hspace{-1.5pt}x}\hspace{1pt}}}

\newcommand{\Honeper}{H^1_\#}
\newcommand{\Honek}{H^1_\kappa}

\newcommand{\half}{{\textstyle{\frac{1}{2}}}}

\newcommand{\e}{\varepsilon}
\newcommand{\ii}{\mathrm{i}}
\newcommand{\io}{{\ii\omega}}
\newcommand{\ik}{{\ii\kappa}}

\newcommand{\xxi}{\xi}

\renewcommand{\|}{{|\hspace{-1.3pt}|\vspace*{-6pt}}}

\newcommand{\pair}[4]{\left#1 \hspace{-0.5em} \begin{array}{l} #2\vspace{0.7ex} \\ #3 \end{array} \hspace{-0.5em} \right#4}
\newcommand{\triple}[5]{\left#1 \hspace{-0.5em} \begin{array}{l} #2\vspace{0.7ex} \\ #3\vspace{0.7ex} \\ #4 \end{array} \hspace{-0.5em} \right#5}

\newcommand{\HH}{\mathbf{H}}
\newcommand{\EE}{\mathbf{E}}

\newcommand{\JJ}{\mathbf{J}}

\newcommand{\Bav}{B_{\mathrm{av}}}
\newcommand{\Hav}{H_{\mathrm{av}}}
\newcommand{\Eav}{E_{\mathrm{av}}}
\newcommand{\Dav}{D_{\mathrm{av}}}
\newcommand{\bav}{b_{\mathrm{av}}}
\newcommand{\hav}{h_{\mathrm{av}}}

\newcommand{\hi}{h_{\mathrm{i}}}
\newcommand{\he}{h_{\mathrm{e}}}
\newcommand{\Ei}{E_{\mathrm{i}}}
\newcommand{\Ee}{E_{\mathrm{e}}}
\newcommand{\vi}{v_{\mathrm{i}}}
\newcommand{\ve}{v_{\mathrm{e}}}
\newcommand{\ui}{u_{\mathrm{i}}}
\newcommand{\ue}{u_{\mathrm{e}}}
\newcommand{\we}{w_{\mathrm{e}}}

\newcommand{\tue}{\tilde u_{\mathrm{e}}}
\newcommand{\tueeta}{\tilde u_{\mathrm{e}}^\eta}
\newcommand{\tuieta}{\tilde u_{\mathrm{i}}^\eta}

\newcommand{\Pe}{P_{\mathrm{e}}}
\renewcommand{\Pi}{P_{\mathrm{i}}}
\renewcommand{\O}{\Omega}
\newcommand{\Oe}{{\Omega_{\mathrm{e}}}}
\newcommand{\Oi}{{\Omega_{\mathrm{i}}}}
\newcommand{\chie}{{\chi_{\mathrm{e}}}}
\newcommand{\chii}{{\chi_{\mathrm{i}}}}
\newcommand{\Oieta}{{\Omega_{\mathrm{i}}^\eta}}
\newcommand{\Oeeta}{{\Omega_{\mathrm{e}}^\eta}}
\renewcommand{\d}{\mathrm{d}}

\newcommand{\sigmas}{\sigma_{\!\mathrm{s}\,}}

\newcommand{\dO}{{\partial\O}}
\newcommand{\dOi}{{\partial\Oi}}
\newcommand{\dOieta}{{\partial\Oieta}}

\newcommand{\D}{{G}}  
\newcommand{\Dc}{{{\D}^*}}
\newcommand{\dD}{{\partial \D}}

\newcommand{\QD}{{{\cal Q}\hspace{-3pt}\setminus\hspace{-3pt}D}}
\newcommand{\QQ}{{\cal Q}}

\newcommand{\kk}{k} 
\renewcommand{\tt}{t} 
\newcommand{\nn}{n} 
\newcommand{\kx}{\kk\!\times\!}
\renewcommand{\dot}{\!\cdot\!}
\newcommand{\strip}{{\cal S}}

\newcommand{\hinc}{h_\mathrm{inc}}
\newcommand{\tsc}{\rightharpoonup\hspace{-1.7ex}\rightharpoonup}
\newcommand{\bigo}{{\cal O}}

\newcommand{\aeo}{a^\eta_\omega}
\newcommand{\beo}{b^\eta_\omega}
\newcommand{\ce}{c^\eta}

\newcommand{\Beo}{B^\eta_\omega}
\newcommand{\Ce}{C^\eta}

\begin{document}

\bibliographystyle{plain}

\title{Magnetism and homogenization of micro-resonators}
\author{Robert V. Kohn\footnote{Courant Institute of Mathematical Sciences, New York University, 251 Mercer Street, New York, NY 10012-1185, USA.
kohn@courant.nyu.edu}
\\
Courant Institute
\and
Stephen P. Shipman\footnote{
Department of Mathematics, Louisiana State University, Lockett Hall 304, Baton Rouge, LA 70803-4918, USA.
shipman@math.lsu.edu
}
\\
Louisiana State University
}
\maketitle

\begin{abstract}\noindent
Arrays of cylindrical metal micro-resonators embedded in a dielectric matrix were proposed by Pendry, {\it et.\,al.}, \cite{PendryHoldenRobbins1999} as a means of creating a microscopic structure that exhibits strong bulk magnetic behavior at frequencies not realized in nature.  This behavior arises for $H$-polarized fields in the {quasi-static} regime, in which the scale of the micro-structure is much smaller than the free-space wavelength of the fields.  We carry out both formal and rigorous two-scale homogenization analyses, paying special attention to the appropriate method of averaging, which does not involve the usual cell averages.  We show that the effective magnetic and dielectric coefficients obtained by means of such averaging characterize a bulk medium that, to leading order, produces the same scattering data as the micro-structured composite.
\end{abstract}


\section{Introduction} 
\label{sec:introduction}

Within the field of artificial materials, there is presently intense activity in the area of creating ``metamaterials" with {\em negative} bulk dielectric or magnetic response.  Materials with dielectric or magnetic coefficients that are either simultaneously negative or of opposite sign offer a rich variety of interesting and useful phenomena.  
As nature provides us with materials that exhibit negative response only at rather restrictive frequencies, one of the aims of this field is to extend the selection of frequencies by creating microscopic structures that have resonant response when natural materials tend to be unresponsive.  
The field received a jump start with the introduction of model structures of thin metallic wires for creating electric resonance \cite{PendryHoldenRobbins1998} and ring-type structures for creating magnetic resonance \cite{PendryHoldenRobbins1999}, which were proposed by Pendry, Holden, Robbins, and Stewart in the late 1990s.  Combinations of these effects were investigated by Smith, {\em et.\,al.} \cite{SmithPadillaVier2000} and many others, to create ``left-handed" materials, possessing a negative index of refraction.  More recently, Pendry \cite{Pendry2004a} proposed the creation of negative refraction by composites in which one of the components is chiral, and the homogenization of such structures has been investigated in \cite{GuenneauZolla2007}.
There seems to be considerable debate and some confusion concerning the definition and meaning of bulk effective electromagnetic coefficients in this setting; moreover, rigorous mathematical treatment of the subject is still in its early stages.

Our intention with this work is to help clarify the meaning of the effective bulk dielectric permittivity and magnetic permeability in the {\em quasi-static limit}, in which the scale of the micro-structure is small compared to the free-space wavelength of the fields.  We work only with a {\em two-dimensional model of ring-type resonators}, for which magnetism is the dominant effect.
The mathematical context of our study is {\em periodic homogenization}.  The relations between the $D$ and $E$ fields and between the $B$ and $H$ fields of the individual components of the micro-structured composite material,
\begin{equation}
  D = \e E,  \qquad B = \mu H,
\end{equation}
give rise to bulk coefficients $\e^*$ and $\mu^*$ that govern certain average fields on the macroscopic level:
\begin{equation}
  \Dav = \e^*\Eav,  \qquad \Bav = \mu^*\Hav.
\end{equation}
What is noteworthy in the homogenization of micro-resonators is that these average fields are not to be understood in the standard way as micro-cell averages.  Indeed, as Pendry, {\it et.\,al.}, \cite{PendryHoldenRobbins1999} observed, even if $\mu = \mu_0$ for all components, we still obtain a nontrivial magnetic response, $\mu^*\not=\mu_0$.

The crucial ingredient for emergence of magnetic behavior is the {\em presence of a component with extreme physical properties}.  In fact, the rings in our resonators must possess high conductivity or internal capacitance tending to infinity as the inverse of characteristic length of the micro-structure.

For broad discussions of electromagnetic materials with negative coefficients, one may consult \cite{Pendry2004}, \cite{SoukoulisKafesakiEconomou2006}, or \cite{Ramakrish2005}, for example.

\subsection{Magnetism from micro-resonators}

In our model, the micro-resonators are represented by infinitely long rods with conducting surfaces (Figure \ref{fig:strip}).  The fields are harmonic (with frequency $\omega$) and magnetically polarized.  The magnetic field, denoted by the scalar $h(x_1,x_2)$, is directed parallel to the rods,
while the electric field $E(x_1,x_2)$ lies in the plane perpendicular to the rods.  The electric field induces a current $j$ on the surfaces of the resonators; it is related to the tangential component $E(x_1,x_2)$ through a {\em complex} (more on this in section \ref{subsec:model}) surface conductivity $\sigmas$: $j = \sigmas E\dot t$.  The current, in turn, effects a discontinuity in the magnetic field, $\hi - \he = j$, where $\he$ and $\hi$ denote the values of $h$ exterior and interior to the micro-resonators.  The Maxwell system of partial differential equations reduces to the following system for $h$ and $E$
(where $\gradp := \kx \grad
         = \left\langle - \partial/\partial {x_2}, \partial/\partial {x_1} \right\rangle
$):
\begin{eqnarray}
  && \pair{.}{\gradp\cdot E - \io\mu h \,=\, 0, }{\gradp h - \io\e E \,=\, 0,}{\}}
             \;\; \text{off the surfaces of the micro-resonators,} \label{Maxwell}\\
  && \he - \hi + \sigmas E\dot t \,=\, 0, \quad \text{with $E\dot t$ continuous on the surfaces,}
              \label{jump}
\end{eqnarray}
with the convention that the unit tangent vector is directed in the counter-clockwise sense.

The quasi-static limit amounts to fixing the frequency and allowing the period of the micro-structure to tend to zero.  This is to be contrasted with work of Sievenpiper, {\it et.\,al.}, \cite{SievenpipYablonoviWinn1998} for example, who devise capacitative structures for the manipulation of photonic spectral gaps, a phenomenon that is pronounced when the wavelength and period are comparable.
Starting from (\ref{Maxwell}--\ref{jump}), we shall derive a system of Maxwell equations governing suitably defined {\em macroscopic} fields.  Because the currents around the resonators flow in microscopic loops, they do not appear as currents in the homogenized equations.  Instead, their bulk effect is manifest through the effective magnetic coefficient $\mu^*$, which is complex (even if $\e$ and $\mu$ are real) because it incorporates the effect of loss due to the currents.  Our homogenized system is
\begin{equation}\label{homogenized1}
  \pair{.}{\gradp\dot \Eav - \ii\omega\mu^* \he^0  \,=\, 0,}
           {\gradp \he^0 - \ii\omega\e^*\Eav \,=\, 0.}{.}
\end{equation}
It is important that the equations for the bulk fields are posed in terms of the {\em exterior value of the magnetic field}, which is denoted by $\he^0$, not in terms of its cell average.

At the risk of redundancy, we emphasize the following two key points:

\begin{enumerate}
\item The characteristic of the micro-structure that is crucial for the emergence of magnetic response from nonmagnetic components is that {\em one of the material properties, the surface conductivity, is extreme.}  More precisely, it scales inversely with the microscopic length scale.
\item In the homogenized Maxwell system, {\em the macroscopic $H$ field is not the cell average of the field but rather the value exterior to the resonator.}  The macroscopic $B$ field, however, is the usual micro-cell average.
\end{enumerate}
In fact, not only are the $H$ and $B$ fields averaged over different parts of a unit cell, but so are the $E$ and $D$ fields.  This means that, just as we have discussed for the magnetic coefficient, even if $\e=\e_1$ in all components of the unit cell, the effective dielectric coefficient $\e^*$ will typically be different from $\e_1$.  This feature
distinguishes the homogenization of micro-resonators from the more ``standard" homogenization of composites with perfectly bonded interfaces and material properties that are not extreme, in which all macroscopic fields are understood as micro-cell averages.

The second feature is already present in the problem of homogenization porous media, in which the cell average is taken {\em outside} the holes \cite{CioranescSaint-Jea1979}.  More recently, both of these features appeared in the work of Bouchitt\'e and Felbacq \cite{BouchitteFelbacq2004,BouchitteFelbacq2005,FelbacqBouchitte2005,FelbacqBouchitte2005a}, who demonstrated the emergence of bulk magnetic behavior from nonmagnetic materials in a somewhat different but related problem.  The extreme property in their setting is the dielectric coefficient inside a periodic inclusion, which tends to infinity as the inverse area of a micro-cell.  The magnetic field in the matrix, exterior to the inclusions, has vanishing fine-scale variation and appears as the macroscopic field in the homogenized equations.  This exterior value drives the fine-scale oscillations in the interior, whose resonant frequencies produce extreme magnetic behavior.  The point in their problem as well as ours is that the effective equation involves only the $H$ field exterior to the inclusion (or resonator), but the $B$ field must be averaged over the entire micro-cell.  A similar problem in which fiber arrays have extreme conducting properties in the plane perpendicular the fibers but not in the direction of the fibers has been investigated rigorously by Cherednichenko, Smyshlyaev, and Zhikov \cite{CherednicSmyshlyaeZhikov2006}.  Here again, the field in the matrix appears in the effective equations, which possess the additional feature of spatial non-locality in the direction of the fibers.
These problems are to be contrasted with the case of a small-volume-fraction array of conducting metallic fibers of finite length, treated by Bouchitt\'e and Felbacq \cite{BouchitteFelbacq2006}.  In that setting, the conductivity is extreme but the field averages are nevertheless taken in the usual way. 

We shall discuss the homogenized system \eqref{homogenized1}---and our scheme for defining ``averaged fields" and ``effective properties"---further in Section \ref{sec:overview}, and we offer a systematic justification in Section \ref{discussion} following the formal asymptotic analysis.  The main points are these:
\begin{enumerate}
\item The effective coefficients $\e^*$ and $\mu^*$ describe the limiting behavior of the scattering problem.
\item Taking the value of $H$ exterior to the resonators as the macroscopic field is the only choice that preserves the Maxwell-type structure of the effective system of PDEs (see the comments following equation \eqref{homogenized0}).
\item The averaging scheme is consistent with the treatment of the $E$, $D$, $H$, and $B$ fields as differential forms (section \ref{discussion}).
\end{enumerate}

Point (1) means that the field scattered by a micro-structured object when illuminated by a plane wave should tend to the field scattered by an object of the same shape consisting of a material possessing the bulk coefficients.  In particular, the reflection and transmission coefficients associated to scattering by a micro-structured slab should approach those for the homogenized slab.  This is the model we have chosen for our analysis.  It is important to keep in mind that, in any scattering problem, we assume that {\em no resonator is cut or exposed to the air}, so that the $H$ field in the matrix dielectric surrounding the resonators connects continuously with the $H$ field in the air.  In other words, the air-composite interface cuts only through the matrix material.  Although we do not treat boundary-value problems, it is evident by the same reasoning that the same effective equations remain valid for problems in which the boundary is exposed only to the matrix.

\subsection{Scaling of fields in the quasi-static limit}

Let us take a more careful look at our particular scalings.

As one expects, the variation of $h$ at the scale of the micro-structure, which we denote by $\eta$, vanishes in the matrix as this fine scale tends to zero (the quasi-static limit); the microscopic variation of $h$ is due solely to its discontinuities at the current-carrying surfaces of the resonators.  (On the other hand, the electric field $E$ will have non-vanishing micro-periodic oscillations in the matrix, even if $\e$ and $\mu$ are constant.)  In the quasi-static limit, therefore, the jump in $h$, which we call the current $j$, is constant around a single micro-resonator.

We have emphasized the point that interesting magnetic behavior arises only when the micro-resonators are highly conducting.  This means that the surface conductivity $\sigmas$ tends to infinity as the size $\eta$ of a period cell tends to zero.
The reason is described by Pendry, {\it et.al.}, \cite{PendryHoldenRobbins1999} and is borne out by our homogenization analysis.
The electromotive force (EMF) around a single resonator is given by the line integral of the electric field around its surface.  If the interior of one resonator occupies a region $\D_\eta$, then, using the first of the Maxwell equations \eqref{Maxwell}, we can write the EMF in two ways:
\begin{eqnarray}
  && \text{EMF} \,=\, \int_{\partial \D_\eta} E\dot t \,\d s \,\sim\, \eta\,\bigo(E\dot t), \\
  && \text{EMF} \,=\, \io\int_{\D_\eta} \mu h \,\d A \,\sim\, \eta^2 \bigo(h).
\end{eqnarray}
As $h$ is of order 1 in $\eta$, $E\cdot t$ is forced to be of order $\eta$ on the surface of the resonators.  Now, using the relation $\he - \hi + \sigmas E\cdot t \,=\, 0$, we observe that, in order that nontrivial behavior emerge at the microscopic level, we should use a highly conducting material, namely, $\sigmas \sim \eta^{-1}$.  In summary, the relevant scalings in our analysis are the following.

\begin{center}
\begin{tabular}{ll}
  $h$ and $E$ in the matrix: & $\bigo(1)$ \\
  current and discontinuity of $h$ on surfaces: & $\bigo(1)$ \\
  $E\dot t$ on surfaces: & $\bigo(\eta)$ \\
  surface conductivity: & $\bigo(\eta^{-1})$
\end{tabular}
\end{center}

\subsection{The model for micro-resonators}
\label{subsec:model}

If the resonator consists of a solid metal cylinder, the connection between the current and the electromagnetic fields is accomplished by a simple constitutive law relating $j$ to $E\cdot t$ through a real scalar $\sigmas$, the surface conductivity.  In this case, the resonator acts purely as an inductor.  If the solid cylinder is replaced by a uniform solid metal ring, effectively the same situation persists, as there is no uneven distribution of charge on the two surfaces of the ring that would give rise to capacitative effects.

More elaborate resonators allow for capacitative effects by forcing nonuniform build-up of charge around the ring.  These are known as split-ring resonators (SRR) (see \cite{PendryHoldenRobbins1999} or \cite{Ramakrish2005}, for example), composite ring structures consisting in part of dielectric material and in part of one or more incomplete metal rings.  As the temporal distribution of charge is out of phase with that of the current, so also are the inductive and capacitative contributions to the jump in~$h$.

It is not our intention to provide a rigorous account of the inductive and capacitative effects of SRRs in the quasi-static limit.\footnote{
The task of providing such an account, {\it i.e.}, justifying equation \eqref{sigmas} below starting from the Maxwell equations, remains in our view an open problem worthy of analysis.
}
Rather, in this work we are content to observe that in the literature, the electromotive force (EMF) around the ring arises from integrating two quantities around it, both related to the current $j$:
\begin{enumerate}
\item An inductive component $\rho j$, where the real quantity $\rho$ is the effective resistance of the metal in the composite ring.
\item A capacitative component $j/(-\io C)$, where the real quantity $C$ is the capacitance arising from splits in the rings or the proximity of two closely stacked metal rings within the compound ring. 
\end{enumerate}

Now, using the representation EMF $= \int_{\partial \D_\eta} E\dot t \,\d s$, we model a general resonator by a single closed loop (a cylindrical surface in the three-dimensional realization)
on which the inductive and capacitative effects are manifest through a single phenomenological {\em complex} constitutive law $j = \sigmas E\dot t$, where $\sigmas = \sigmas^1 + \ii\sigmas^2$.  This amounts to defining a singular current $\sigmas^1 E\cdot t$ and a singular $D$ field of strength $-(\sigmas^2/\omega) E\cdot t$ around the resonator.

We now demonstrate the nature of the correspondence between our formula for the effective magnetic permeability and that obtained by Pendry, {\it et.\,al.}  First, consider the case that $\mu=\mu_0$ in all components and that the resonator is represented by a simple metal cylinder of (nondimensional) radius $R<1$ in relation to a scaled unit cell (Fig. \ref{fig:cell}, right).  The actual radius is $r = \eta LR$, where $L$ is an arbitrary fixed length corresponding to 1 in the macroscopic variable $x$.  Keeping in mind that the conductivity should tend to infinity with $\eta^{-1}$, we set it equal to
\begin{equation}
\sigmas = \frac{1}{\eta\,\rho\,}\,,
\end{equation}
where $\rho$ is a fixed real constant.  From equations \eqref{one}, \eqref{two}, and \eqref{three}, we obtain
\begin{equation} \label{mustar}
\mu^* = \mu_0\left[ 1 - \pi R^2 \left( 1 + \frac{2\ii\rho}{\omega R\mu_0} \right)^{\!-1} \right],
\end{equation}
which is the formula (13) in \cite{PendryHoldenRobbins1999}.\footnote
{To make the connection to formula (13) $\mu_\mathrm{eff} = 1-\frac{\pi r^2}{a^2}\left( 1 + \frac{2\ii\sigma}{\omega r\mu_0} \right)^{-1}$ and (17) or (20)
$\mu_\mathrm{eff} = 1-\frac{\pi r^2}{a^2}\left( 1 + \frac{2\ii\sigma}{\omega r\mu_0} - \frac{3d}{\pi^2\mu_0\omega^2\e_0r^3}\right)^{-1}$ in \cite{PendryHoldenRobbins1999}, one must relate the notation in that paper to our notation in the following way: $a\mapsto\eta L$, $r\mapsto r$, $d\mapsto d$, $\mu_0\mu_\mathrm{eff}\mapsto\mu^*$, $\sigma\mapsto\eta L\rho$.
Note that the symbol $\sigma$ in \cite{PendryHoldenRobbins1999} denotes the resistance of the metal, which is the reciprocal of conductivity.
}

Now, in order to incorporate capacitance into a SRR as in Figure \ref{fig:cell}, we take a phenomenological step and observe that, in order that our model produce formula (20) of \cite{PendryHoldenRobbins1999}, the {\em complex} conductivity must now be set equal to
\begin{equation}\label{sigmas}
  \sigmas = \frac{1}{\eta}\left( \rho + \ii\tau \right)^{-1} = \frac{1}{\eta}\left( \rho + \ii\frac{3\Delta}{2\pi^2 \omega \e_0 r^2} \right)^{\!-1},
\end{equation}
where $d=\eta L\Delta$ is the small distance between the outer and inner shells making up the circumference of the resonator.  Here, $\Delta$ is the nondimensionalized distance in a rescaled unit cell (Fig. \ref{fig:cell}, left).  The coefficient $\e_0$ appears in \ref{sigmas} because it is assumed that $\e=\e_0$ in the dielectric between the inner and outer rings.  The formula for $\mu^*$ becomes
\begin{equation}\label{mustar1}
\mu^* = \mu_0\left[ 1 - \pi R^2 \left( 1 + 
                          \frac{2}{\omega R\mu_0}\left( \ii\rho - \tau \right) \right)^{\!-1} \right].
\end{equation}

In order that the imaginary part of $\sigmas$ tend to infinity as $\eta^{-1}$, we must have
\begin{equation}
  \Delta \sim \eta^2.
\end{equation}
This means that the actual distance $d=\eta L\Delta$ between shells would have to decrease as the {\em cube} of the cell length in order for the complex part of $\sigmas$ to have an effect in formula~\eqref{mustar} in the limit as $\eta\to0$.  Of course, if the dielectric coefficient is allowed to tend to zero at some rate, the extreme rate of convergence of $\Delta$ to zero can be relaxed.

We observe from \eqref{mustar1} that if $\mu_0$ is real and positive then $\mu^*$ has positive imaginary part (provided $\rho>0$).  The assumption used above that $\mu\equiv\mu_0$ in all components was merely a simplification; even if $\mu$ varies in space, a similar calculation of $\mu^*$ is straightforward.  It reveals that, if $\mu$ is real and positive throughout the structure, then
\begin{equation}\label{Immu}
  \Im(\mu^*) > 0.
\end{equation}
We shall need this fact in Section \ref{sec:twoscale}.

We note that our homogenization-based analysis might not always be adequate for modeling the behavior of a specific device.  In particular, it is difficult to fabricate micro-resonators with shells that are extremely close to each other; moreover, the regime of interest is often not quasi-static, {\it e.g.}, the wavelength may be only several times the length of a unit cell.

\begin{figure}
\begin{center}
\includegraphics[width=0.90\textwidth]{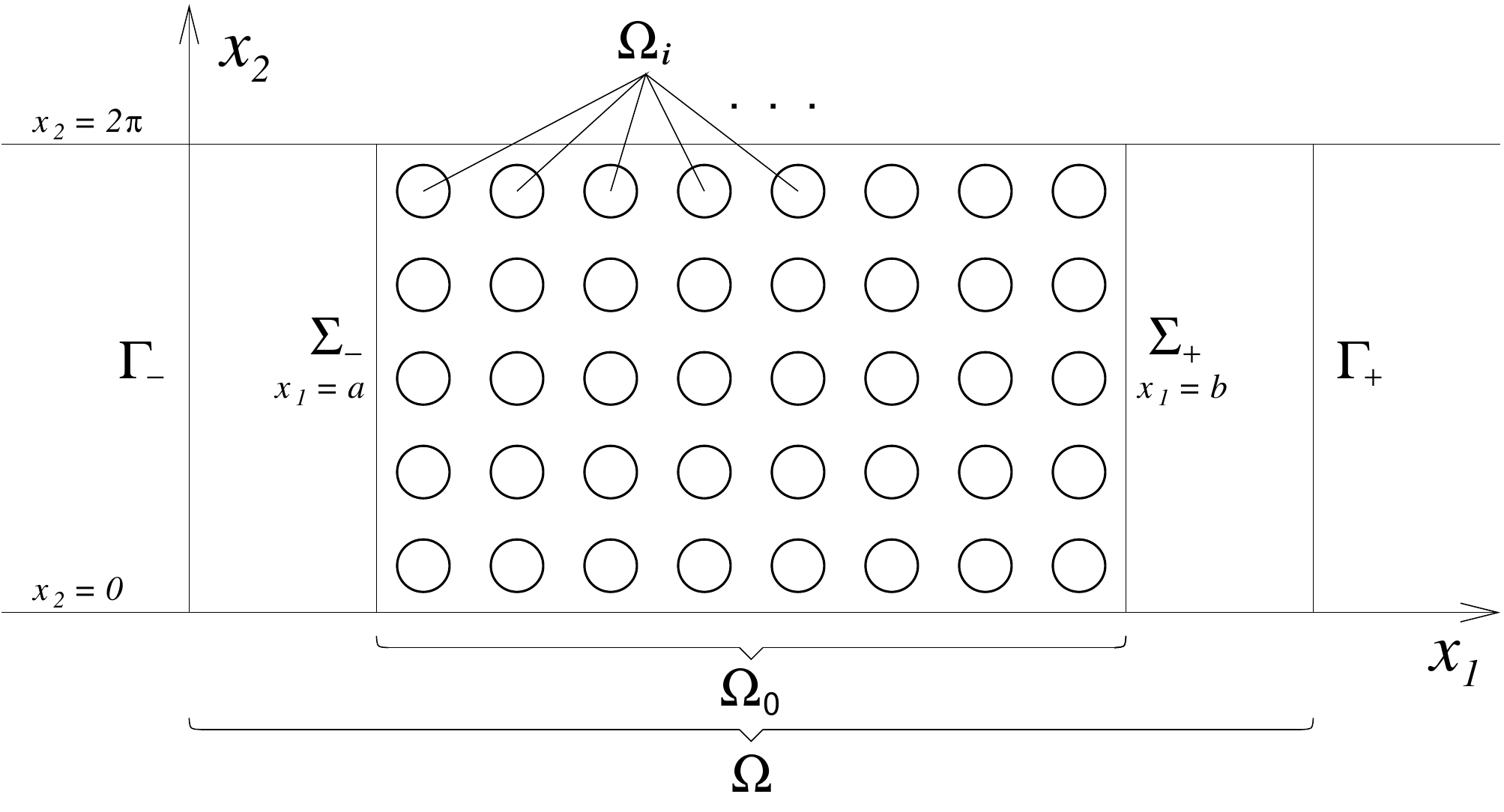}
\end{center}
\caption{The strip $\cal S$, consisting of one period in $x_2$ of the micro-structured slab and  surrounding air.  The segments $\Gamma_\pm$ are artificial boundaries used in the weak formulation of the scattering problem.  We define also $\Oe= \Omega\setminus\Oi$ and
$\Omega_{0\text{e}} = \Omega_0 \setminus \Oi$.}
\label{fig:strip}
\end{figure}

\begin{figure}
\begin{center}
\includegraphics[width=0.3\textwidth]{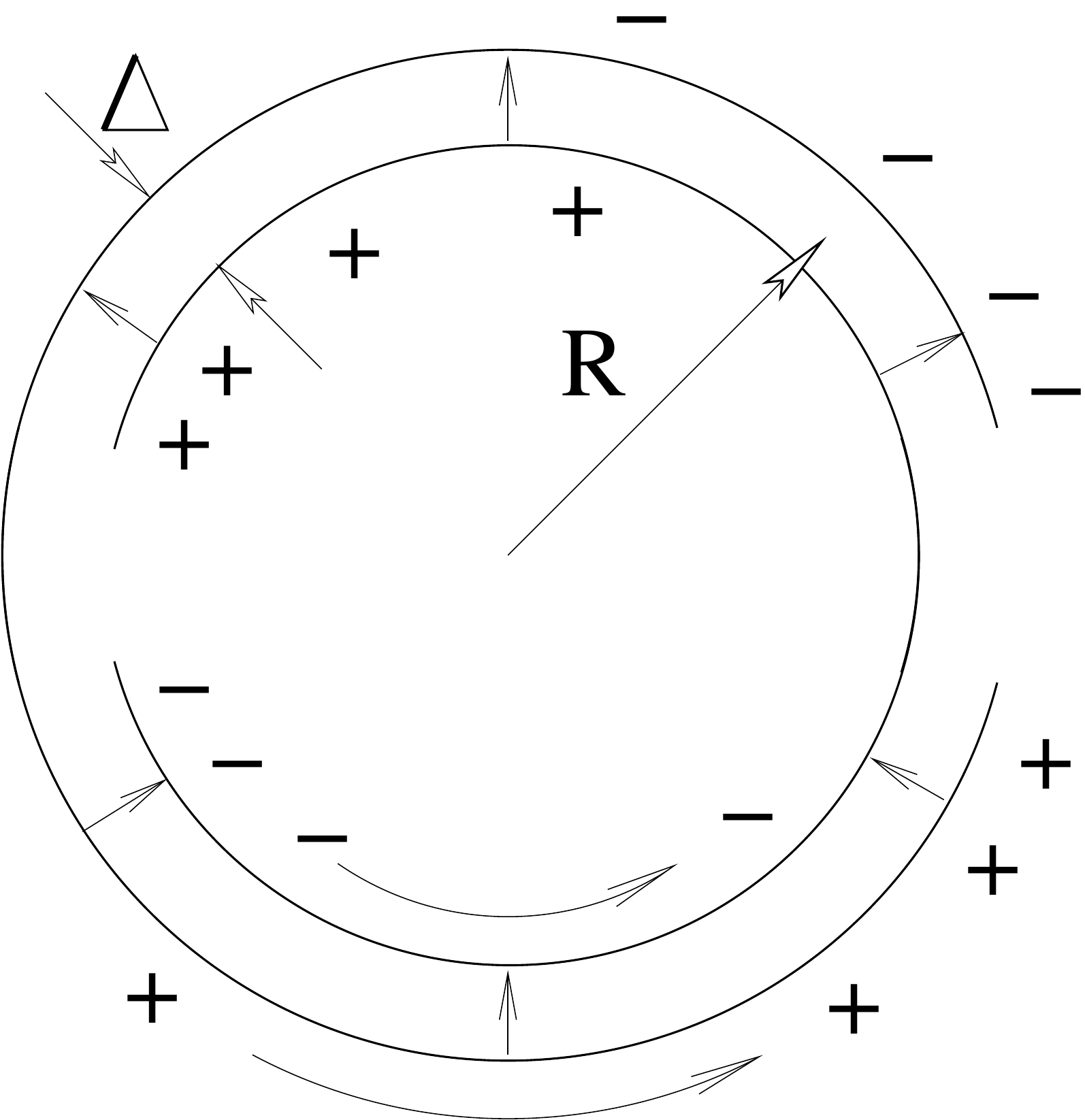}
\hspace{0.15\textwidth}
\setlength{\unitlength}{1.2em}
\begin{picture}(10,10)
  \thinlines
  \put(1,1){\vector(1,0){9}}
  \put(1,1){\vector(0,1){9}}
  \put(1,8){\line(1,0){7}}
  \put(8,1){\line(0,1){7}}
  \put(0.8,0){0}
  \put(0,0.8){0}
  \put(7.8,0){1}
  \put(0,7.8){1}
  \put(1.3,9.1){$y_2$}
  \put(8.9,1.5){$y_1$}
  \thicklines
  \put(4.5,4.5){\circle{4}}
  \put(4.15,4.15){$\D$}
  \put(6.3,6.3){$\D^*$}
  \put(1.4,6.9){$\partial\D$}
  \thinlines
  \put(2.2,6.6){\vector(1,-1){0.9}}
\end{picture}
\end{center}
\caption{{\bf Left:} An example of a split-ring resonator.  The current flows in the direction of the arrows.  The splits in the rings cause the charge to be nonuniform, creating a large $D$-field between the rings.  This allows the current to follow a complete circuit by passing from one ring to the other through a capacitative gap.  {\bf Right:}  The microscopic unit cell $\QQ$ coordinatized by the variable $y= x/\eta$.  The idealized metal-dielectric micro-resonator consists of the boundary $\dD$ of a simply connected domain $\D$.}
\label{fig:cell}
\end{figure}
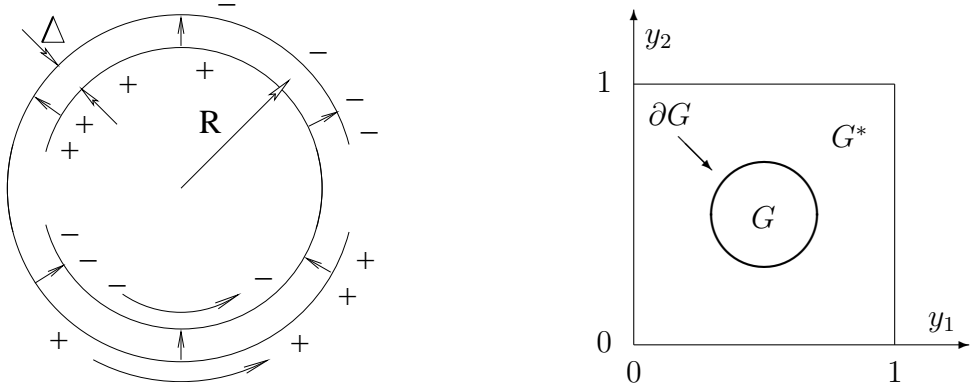

\subsection{Our approach}

The effective tensors $\e^*$ and $\mu^*$ and the unit-cell problem that determines them, together with the first-order correction to the leading order $H$ field, are obtainable through formal asymptotic analysis, which we carry out because we feel it provides a clear intuitive point of view.  We then prove the main results rigorously using the method of two-scale convergence (see Allaire \cite{Allaire1992} for a systematic treatment).

The main ideas of the rigorous two-scale arguments are these.  The weak form of the scattering problem, Problem \ref{microweak} in Section \ref{sec:twoscale}), is posed within the $\eta$-dependent space
\begin{equation}
H^1(\Oeeta)\oplus H^1(\Oieta),
\end{equation}
where $\Oeeta$ and $\Oieta$ are the domains exterior to and interior to the micro-resonators.  The boundary term, involving the conductivity, brings the jump discontinuity of the $H$ field into the equation.  Notice that, if the conductivity remained bounded as $\eta\to0$, then this term would be of order $\eta^{-1}$ and the jump would disappear in the limit.  Thus, as we have pointed out, we take the conductivity to be of order $\eta^{-1}$ so that the boundary integrals remain on the same order as the the area integrals.

We then show that the scattering problem always has a solution $h^\eta$ for each $\eta$.  As we do not have {\it a priori} bounds on the solutions, we first scale them and obtain the two-scale limits of the scaled solutions as well as those of their gradients.  The uniqueness of the solution of the homogenized system governing these scaled solutions then allows us to obtain {\it a postiori} uniform bounds on the actual solutions.

The main result is Theorem \ref{thm:vartwoscale}, which presents the two-scale variational problem for the functions $\he^0(x)$, $\hi^0(x)$, $\he^1(x,y)$, and $\hi^1(x,y)$ in the two-scale limits
\begin{equation}
\pair{.}{h^\eta(x) \tsc h^0(x,y) = \chie(y)\he^0(x) + \chii(y)\hi^0(x), \vspace{3pt}}
           {\grad h^\eta(x) \tsc \chie(y)[\grad \he^0(x) + \grady \he^1(x,y)]
                                                      + \chii(y)[\grad \hi^0(x) + \grady \hi^1(x,y)], \vspace{3pt}}
           {.}
\end{equation}
and its equivalence to the homogenized Maxwell system together with the unit cell problem (section \ref{subsec:cell}) that determines the gradients of the corrector functions $\he^1$ and $\hi^1$ as well as the effective coefficients $\mu^*$ and $\e^*$.

In Theorem \ref{thm:strong}, we obtain the strong two-scale convergence of $h^\eta$ and $\grad h^\eta$ to their two-scale limits.  Strong two-scale convergence strengthens two-scale convergence by asserting convergence of energies.
Theorem \ref{thm:transmission} asserts the convergence of the transmission and reflection by the micro-structured slab to that of the homogenized one.

\bigskip
\noindent
{\bf \large Acknowledgments}
\medskip

\nopagebreak
The authors would like to acknowledge the inspiration of the IMA ``Hot Topics" Workshop on
Negative Index Materials in October of 2006 at the University of Minnesota, cosponsored by the Air Force Office of Scientific Research.  This Workshop provided the foundation and stimulation for our work in this area.

The effort of R.V.K. was supported by NSF grants DMS-0313890 and DMS-0313744, and the effort of S.P.S was supported by NSF grant DMS-0313890 during his visit at the Courant Institute and by NSF grant DMS-0505833.

\section{Overview of results}  
\label{sec:overview}

This section serves to highlight the main results of the calculations of Section \ref{sec:expansion}, including the ``unit-cell problem" for the corrector functions and definitions of the effective electric and magnetic coefficients.  In the presentation, we try to illuminate the physical meaning of the results and the methods of averaging.

The feature of our system that leads to interesting magnetic behavior, as previously discussed, is the high surface conductivity of the resonators, which scales inversely to the length of the period micro-cell, that is, $\sigmas(x) = \sigma(x,x/\eta)/\eta$ for $x$ on the surfaces of the resonators.  The dielectric and magnetic coefficients in the matrix exterior and interior to the surface of the resonators remain bounded, and are given by $\e(x,x/\eta)$ and $\mu(x,x/\eta)$.

The Maxwell system for $H$-polarized fields reduces to an elliptic equation for the scalar magnetic field $h$ which includes interaction with a current on the surface of a conducting resonator.
Because of the extreme value of the conductivity, to leading order, the magnetic field $h$ behaves on the microscopic scale as a piecewise constant function, with jump discontinuities at the micro-resonator interfaces given by the surface current.  The electric field, by contrast, exhibits fine-scale oscillations of order 1 in the matrix, and its tangential component is continuous across the micro-resonator interfaces.

The exterior and interior values of $h$, as functions of the macroscopic variable $x$, are related to each other through a fixed complex ratio $m(x)$ determined by the shape and conductivity of the micro-resonator, the magnetic coefficient in its interior, and the frequency.  Representing the leading-order behavior of the exterior and interior values of $h$ by the functions $\he^0(x)$ and $\hi^0(x)$, we will show in section \ref{subsec:cell} that
\begin{equation}
\hi^0(x) = m(x) \he^0(x), \qquad   m(x) = \frac{\hat\rho(x)}{\hat\rho(x) - \io\hat\mu(x)},
\end{equation}
in which $\hat\mu$ and $\hat\rho$ are weighted averages of the magnetic coefficient and the resistance (the reciprocal of the conductivity) of the micro-resonator.
%
%
Thus as a function of macro- and microscopic variables, the leading behavior of $h$ is given by
\begin{equation}
  h^0(x,y) = M(x,y) \he^0(x),
\end{equation}
where $M(x,y)$, which characterizes the microscopic variation of $h$, takes on two values:
$M(x,y) = \chi_\D(y) + m(x)\chi_\Dc(y)$.  As we have mentioned in the introduction, the physical principle behind expression for $m(x)$ is a balance of electromotive forces, which, in the microscopic variable, takes the form
\begin{equation}\
  \int_\dD \frac{1}{\sigma(x,y)} \left( \hi^0(x)-\he^0(x) \right) \,\d s(y) \,=\, \io\! \int_\D \mu(x,y) \hi^0(x) \,\d A(y)\,.
\end{equation}

The question now arises as to which value of $h$ should be recognized as the appropriate macroscopic magnetic field in the homogenized bulk medium.  The model problem of scattering by a slab that we have chosen represents a practical class of problems in which the boundary of the composite cuts through the dielectric matrix only so that the boundary conditions involve only the exterior $H$ field.  In our case, we have interface conditions between the slab and the surrounding air (continuity of $h$ and $E\dot t$).
Therefore we arrive at effective equations that govern $\he^0(x)$ and the usual micro-cell average $\Eav(x)$ of the $E$ field:
\begin{equation}\label{homogenized0}
  \pair{.}{\gradp \dot \Eav(x) - \ii\omega\mu^*(x) \he^0(x)  \,=\, 0,}
           {\gradp \he^0(x) - \ii\omega\e^*(x)\Eav(x) \,=\, 0.}{.}
  \qquad \text{(homogenized system)}
\end{equation}

This system is evidently of two-dimensional Maxwell form, which may be seen {\it a postiori} as another justification for choosing the exterior value of $h$.  Indeed, if the micro-resonators vary macroscopically, implying that $m(x)$ is not constant, then the system \eqref{homogenized0}, rewritten in terms of $\hi^0$ or some combination of $\he^0$ and $\hi^0$, say $\hav^0(x) = a(x)\he^0(x)$, would contain an extra term $-a(x)^{-1}\gradp a(x) \hav^0(x)$ in the second equation of \eqref{homogenized0}, placing the system outside of the usual Maxwell type.  The reason that this term is not present with the choice $\he^0$ lies in the unit cell problem for the corrector functions.

Because to leading order $E\dot t$ vanishes on the surfaces of the micro-resonators, the cell problem is decoupled into parts exterior and interior to the domain $G$.  It gives the order-$\eta$ corrector $h^1(x,y)$ as well as the leading-order part of the $E$ field, $E^0(x,y)$.  In its weak form, it reads
\begin{equation}    \label{cellext}
  \pair{.}{\displaystyle\int_\Dc \Ee^0(x,y) \dot \gradp v(y) \,\d A(y) = 0
       \quad \text{for all} \; v\in \Honeper(\Dc),\vspace{1ex}}
      {\Ee^0(x,y) = \frac{1}{\io} \e(x,y)^{-1} \left( \gradp \he^0(x) + \gradpy \he^1(x,y) \right),}
      {.}   \qquad  \text{(exterior cell problem)}
  \end{equation}
in the exterior, which is a periodic-Neumann boundary-value problem for $\he^1(x,y)$ and
\begin{equation}  \label{cellint}
  \pair{.}{\displaystyle\int_\D \Ei^0(x,y) \dot \gradp v(y) \,\d A(y) = 0
       \quad \text{for all} \; v\in H^1(\D),\vspace{1ex}}
      {\Ei^0(x,y) = \frac{1}{\io} \e(x,y)^{-1} \left( \gradp \hi^0(x) + \gradpy \hi^1(x,y) \right),}
      {.}  \qquad  \text{(interior cell problem)}
  \end{equation}
in the interior, which is a Neumann problem for $\hi^1(x,y)$.  The interior cell problem has a simple explicit solution:
\begin{equation}
  \pair{.}{\gradpy \hi^1(x,y) = -\gradp\hi^0(x),}
             {\Ei^0(x,y) = 0.}
     {.} \qquad \text{(interior)}
\end{equation}

The second equation of \eqref{homogenized0} relates $\he^0$ to $\Eav(x) = \int_\QQ E^0(x,y) \d A(y)$.  It is obtained by averaging $E^0(x,y)$ over $\QQ$, which is equivalent to integrating $\Ee^0(x,y)$ over $\Dc$ (since $\Ei^0=0$), 
and using the definition in the second equation in \eqref{cellext}.  Thus $\Eav(x)$ is a linear function of $\gradp \he^0(x)$:
\begin{equation}\label{Eav0}
  \Eav(x) = \frac{1}{\io} \int_\Dc \e(x,y)^{-1}\left( \gradp \hi^0(x) + \gradpy \hi^1(x,y) \right)\d A(x)
     = \frac{1}{\io} \e^*(x)^{-1} \gradp\he^0(x).
\end{equation}
This expression defines $\e^*(x)$ in the usual way as an effective tensor for a periodic medium with inclusions that are perfect conductors, inside of which the electric field vanishes. 
In Section \ref{discussion}, we argue that $\e^*$ is an effective dielectric permittivity in the sense that $\Dav = \e^*\Eav$, where the average electric displacement field $\Dav$ is defined in the appropriate way.

%
%

The effective magnetic permeability $\mu^*(x)$ is defined by
\begin{equation}
  \mu^*(x) = \int_\QQ M(x,y) \mu(x,y) \d A(y),
\end{equation}
and arises upon integrating over $\QQ$ the equation
\begin{equation}
  \gradpx\dot E^0(x,y) + \gradpy\dot E^1(x,y) - \io\,\mu(x,y) h^0(x,y) = 0, \quad \text{for $y\notin\dD$},
\end{equation}
where $E^1$ is the order-$\eta$ corrector to $E^0$.

The noteworthy feature in this definition is that, even if $\mu$ is real (and, say, constant throughout $\QQ$), $\mu^*$ will still be complex-valued.  This is because the usual cell average is used in averaging the $B$ field $\mu h$ whereas the average of $h$ is taken exterior   to the micro-structure.
The physical correctness of this method of averaging was noticed by Pendry, {\em et.\,al.}.  In section \ref{discussion}, we give a discussion of the averages of the four fields $E$, $h$, $D$, and $b$ and how they emerge naturally in the effective equations.

Finally, we come to the specific problem of the scattering of plane waves by a slab of the composite medium.  Of course, one expects that in the quasi-static limit, the reflection and transmission coefficients will approach those of a slab of a homogeneous medium possessing the effective coefficients $\e^*$ and $\mu^*$.  In fact, in practice one would like to be able to {\em deduce} these bulk properties from scattering measurements.  This has been the subject of many works in the literature, for example those of O'Brien and Pendry \cite{OBrienPendry2002,OBrienPendry2002a} and Smith, {\it et.\,al.}, \cite{SmithSchultzMarko2002}, in which the scattering data obtained from a transfer matrix method for a periodic slab are compared to those obtained from expressions for the effective tensors of the corresponding homogeneous medium.  In the present work, this agreement is established in the quasi-static limit by demonstrating that the limiting values of the electric and magnetic fields to the left and right of the slab are precisely those associated with scattering by a homogeneous slab.  This result is embodied in the statement that the limiting electromagnetic fields outside the slab connect continuously to the fields $\Eav$ and $\he^0$ within the slab, where they solve the effective equations~\eqref{homogenized0}.

\section{Scattering by a slab with micro-resonators}  

As our model problem, we choose the scattering of plane waves by a micro-structured slab of a composite material formed by periodically embedding metal micro-resonators in a dielectric matrix.  The width of the slab is fixed, as the size of a cell tends to zero.  We allow the material properties as well as the shape of the micro-resonators to vary on the macroscopic scale, but require that this variation be periodic with period $2\pi$ in the $x_2$ direction, parallel to the slab.  The fields will then be allowed to be pseudo-periodic in $x_2$.

As we have discussed, an important feature of the model is the requirement that, at the slab-air interface, only the dielectric matrix have contact with the air.

\subsection{Reduced Maxwell equations}

Due to the macroscopic periodicity of the structure and the pseudo-periodicity of the fields, the analysis of the scattering problem can be restricted to the strip
\begin{equation}
  \strip := \left\{ (x_1,x_2) : 0\leq x_2 \leq 2\pi \right\}.
\end{equation}
As illustrated in Figure \ref{fig:strip}, a period of the slab occupies a fixed subdomain $\Omega_0$ of $\strip$, bounded below and above by the lines $x_2=0$ and $x_2=2\pi$ and on the sides by $\Sigma_-$ (at $x_1=a$) and $\Sigma_+$ (at $x_1=b$).  The union of the regions enclosed by the micro-resonators contained in $\Omega_0$ is called the interior domain $\Oi$, and its boundary $\dOi$ is identified with the resonators themselves, or, perhaps more correctly, the surfaces of the resonators.

Because we consider electromagnetic fields in a linear medium, we may treat harmonic fields at fixed frequency $\omega$,
\begin{equation}
  \boldsymbol{\cal E}(x_1,x_2,x_3,t) = \EE(x_1,x_2,x_3)e^{-\io t}, \quad
  \boldsymbol{\cal H}(x_1,x_2,x_3,t) = \HH(x_1,x_2,x_3)e^{-\io t},
\end{equation}
which, when inserted into the time-dependent Maxwell system, yield the harmonic Maxwell system for the spatial envelopes of the fields,
\begin{equation}
  \pair{.}{\curl \EE - \io\, \mu \HH = 0,}
          {\curl \HH + \io\, \e \EE - \JJ = 0.}{.}
\end{equation}
In this work, we study two-dimensional structures, in which the medium and the fields are invariant in one spatial dimension.  In this situation, the Maxwell equations decouple into $E$-polarized and $H$-polarized fields; we study the latter:
\begin{equation}
  \triple{.}{\HH(x_1,x_2,x_3) = \langle 0,0,h(x_1,x_2) \rangle,}
            {\EE(x_1,x_2,x_3) = \langle E_1(x_1,x_2),E_2(x_1,x_2),0 \rangle,
                \; E(x_1,x_2) = \langle E_1(x_1,x_2),E_2(x_1,x_2) \rangle,}
            {\JJ(x_1,x_2,x_3) = \langle J_1(x_1,x_2),J_2(x_1,x_2),0 \rangle,
                \; J(x_1,x_2) = \langle J_1(x_1,x_2),J_2(x_1,x_2) \rangle.}{.}
\end{equation}
With the notation
\begin{equation} 
\gradp := \kx \grad
         = \left\langle -\frac{\partial}{\partial {x_2}}, \frac{\partial}{\partial {x_1}} \right\rangle
\end{equation}
(the formal adjoint of $\gradp\cdot$\, is $-\gradp$), the curls of $\EE$ and $\HH$ are expressed as
\begin{equation}
  \pair{.}{\curl\EE = \gradp\dot E\,,}
          {\curl\HH = -\gradp h\,,}{.}
\end{equation}
and the Maxwell system reduces to
\begin{equation}
  \pair{.}{\gradp\dot E - \io\,\mu h = 0,}
          {\gradp h - \io\,\e E + J = 0.}{.}
\end{equation}
The relevant fundamental theorem involving the operator $\gradp\cdot\,$ is
\begin{equation}
\int_R \gradp\dot F \,\d A = \int_{\partial R} F\dot t \,\d s \,,
\end{equation}
in which $F$ is a vector field, $R$ is a bounded domain in $\mathbb{R}^2$ with boundary $\partial R$, and $t$ is the counter-clockwise unit tangent vector to $\partial R$.

We incorporate the regular part of the current $J$ into the electric displacement field $D=\e E$ by allowing $\e$ to be complex\footnote 
{With $\e = \e' + \ii\e''$, the conductivity in the material away from the metal rings is $\sigma := \omega\e''$.  In $-\io\,\e E = -\io\,\e' E + \sigma E$, the term $-\io\,\e' E$ is identified with the time derivative of the electric displacement and $\sigma E$ with the current density $J$ off the surface of the resonator. 
},
and retain only the singular part of $J$ in the equation.  As explained in section \ref{subsec:model}, our resonators are characterized by a complex surface conductivity $\sigmas$ relating the surface current to the electric field on the boundary of each resonator:
\begin{equation}\label{current}
  J = \sigmas (E\dot \tt)\delta_\dOi \tt\,.
\end{equation}
%

Fixing a Bloch wave vector $\kappa$ in the $x_2$ direction in the first Brillouin zone, $\kappa \in [-1/2,1/2)$, the scattered magnetic field $h_\mathrm{sc}$ (total field minus incident field) has the outgoing form
\begin{equation}\label{outgoing}
  h_\mathrm{sc}(x_1,x_2) = \sum_{m=-\infty}^{\infty} c_m^\pm e^{\ii((m+\kappa) x_2 + \nu_m |x_1|)}, 
  \quad \text{$x_1<a$ or $x_1>b$},
\end{equation}
in which the exponents $\nu_m$, which depend on $\omega$ and $\kappa$, are defined by
\begin{equation}
  \nu_m^2 + (m+\kappa)^2 - \e_0\mu_0 \omega^2 = 0
\end{equation}
and the convention that $\nu_m > 0$ if $\nu_m^2 > 0$ and $\ii\nu_m < 0$ if $\nu_m^2 < 0$.
The incident magnetic field is
\begin{equation}
  \hinc(x_1,x_2) = e^{\ii(\bar m+\kappa)x_2} e^{\ii\nu_{\bar m} x_1},
\end{equation}
in which $\nu_{\bar m} > 0$.

\subsection{Microstructure}

The micro-structure at each value of $x$ is described by means of a microscopic variable $y\in\mathbb{R}^2$  and its fundamental period cube $\QQ=[0,1]^2$ (Fig. \ref{fig:cell}).
The boundaries of the micro-resonators are defined through
a microscopic domain $\D$ in $\QQ$ with boundary $\dD$ and exterior domain $\Dc = \QQ\setminus\D$ and three complex-valued functions of both macroscopic and microscopic variables that have period $\QQ$ in $y$.  The domain $\D$ is assumed to be simply connected and contained wholly within the unit cell $\QQ$.  Our analysis could be extended to the case in which $\D$ consists of more than one simply connected component or in which the micro-resonators cannot be modeled by a domain contained wholly within the unit cell.  However, our model seems to include all of the split-ring resonators we have seen in the literature.

The dielectric permeability $\e(x,y)$ is a tensor, whereas the magnetic permeability $\mu(x,y)$ and the conductivity $\sigma(x,y)$ are scalars.  The tensor $\e$, considered as a matrix, is symmetric, the real parts of all three quantities are positive and bounded from above and below, and their imaginary parts are semidefinite:
\begin{equation}
\triple{.}{0 < \e_- \leq \Re\e(x,y)\xxi \dot \xxi \leq \e_+, \quad 0\leq \Im\e(x,y)\xxi\dot\xxi, \quad x\in\Omega_0,\, y\in\QQ, \, \xxi \in \mathbb{R}^2,}
              {0 < \mu_- \leq \Re\mu(x,y) \leq \mu_+, \quad 0\leq \Im\mu(x,y), \quad x\in\Omega_0,\, y\in\QQ,}
              {0 < \sigma_{\!-} \leq \Re\sigma(x,y) \leq \sigma_+, \quad \Im\sigma(x,y) \leq 0, \quad x\in\Omega_0,\, y\in\dD.}
              {.}
\end{equation}
Both $\e$ and $\mu$ are assumed to be continuous in $x$, and continuous in $y$ off the boundary $\dD$.   Thus we allow different values in the interior and exterior of the resonators.  To define the actual micro-structure at a fixed fine scale $\eta$, we set
\begin{equation}
\triple{.}{\e^\eta(x) = \e(x,x/\eta), \; x\in\Omega_0,}
           {\mu^\eta(x) = \mu(x,x/\eta), \; x\in\Omega_0,}
           {\sigma^\eta(x) = \sigma(x,x/\eta), \; x\in\dOi.}
           {.}
\end{equation}
As we wish to allow the conductivity of the metal cylinders to tend to infinity with $1/\eta$, we take
\begin{equation}
\sigmas = \frac{1}{\eta}\,\sigma^\eta(x)
\end{equation}
as the surface conductivity in equation \eqref{current}.  Outside the slab, we take $\e$ and $\mu$ to be constant:
\begin{equation}
 \e^\eta(x) = \e_0 \quad \text{and} \quad \mu^\eta(x) = \mu_0
 \qquad \text{for $x\in\Omega\!\setminus\!\Omega_0$}.
\end{equation}
The interior domain $\Oi$ depends on $\eta$ and is expressed in terms of the domains $\D$ through
\begin{equation}
\Oieta = \{ x\in\Omega_0 : x/\eta \in \D + n \text{ for some $n\in\mathbb{Z}^2$}\}.
\end{equation}

We have discussed the restriction that the edges of the slab not intersect the resonators.  It is convenient, however, to assume a bit more: that the width of the slab encompass an integer number of period cells.  Therefore we require $\eta = (b-a)/n$ for some integer $n$.  Since we also assume that the structure is $2\pi$-periodic in the slow variable $x_2$, we also require that $\eta = 2\pi/m$ for some integer $m$.  These conditions are equivalent to the condition that $b-a$ is rationally related to $\pi$.
The set of permissible values of $\eta$ is denoted by $\Upsilon$:
\begin{equation}
\eta \in \Upsilon.
\end{equation}

For each $\eta$, we let $E^\eta$, $h^\eta$ denote a solution to the scattering problem, which is posed in its strong PDE form as follows.  The subscripts $e$ and $i$ refer to exterior and interior values of functions.

\begin{problem}[Scattering by a slab, strong form] 
\label{strong}
Find a function $h^\eta$ and a vector field $E^\eta$ in the strip $\strip$ such that
\begin{eqnarray}
&& \pair{.}{\gradp\dot E^\eta - \io\,\mu^\eta h^\eta = 0}
                  {\gradp h^\eta - \io\,\e^\eta E^\eta = 0}{\}} \; \text{on $\strip\setminus\dOieta$}, \\
&& E^\eta\dot \tt \; \text{continuous across $\dOieta$}, \\
&& \he^\eta - \hi^\eta + \frac{\sigma^\eta}{\eta} E^\eta\dot \tt = 0 \; \text{on $\dOieta$}, \\
&& h^\eta(x_1,2\pi) = e^{\ik x_2} h^\eta(x_1,0) \; \text{ and } \; \partial_{x_2} h^\eta(x_1,2\pi) = e^{\ik x_2} \partial_{x_2} h^\eta(x_1,0), \\
&& h^\eta(x) = \hinc(x_1,x_2) + \sum_{m=-\infty}^{\infty} a_m e^{\ii((m+\kappa) x_2-\nu_m x_1)} \; \text{for $x_1<a$}, \\
&& h^\eta(x) = \sum_{m=-\infty}^{\infty} b_m e^{\ii((m+\kappa) x_2 + \nu_m x_1)} \; \text{for $x_1>b$}.
\end{eqnarray}
\end{problem}
Equivalently, one could pose the PDE as a divergence-form elliptic operator on $h^\eta$ alone and use the second Maxwell equation as the definition of $E^\eta$:
\begin{equation}
  \pair{.}{\gradp\dot((\e^\eta)^{-1}\gradp h^\eta) + \omega^2\,\mu^\eta h^\eta = 0}
          {E^\eta = \frac{1}{\io}\,(\e^\eta)^{-1}\gradp h^\eta}{\}}  \; \text{in $\strip\setminus\dOieta$}.
\end{equation}
The first equation reduces to $\grad\dot((\e^\eta)^{-1}\grad h^\eta) + \omega^2\,\mu^\eta h^\eta = 0$ if $\e$ is scalar, and conditions involving $E^\eta\dot t$ are parsed in terms of $h^\eta$ through $E^\eta\dot \tt = \frac{1}{\io} (\kx \e^{-1}\gradp h^\eta) \dot \nn$\,.

\section{Formal expansion analysis} 
\label{sec:expansion}

In the formal asymptotic analysis of the strong form system, we assume that $h^\eta$ and $E^\eta$ have expansions
\begin{eqnarray}
  && h^\eta(x) = h^0(x,x/\eta) + \eta h^1(x,x/\eta) + \bigo(\eta^2), \\
  && E^\eta(x) = E^0(x,x/\eta) + \eta E^1(x,x/\eta) + \bigo(\eta^2),
\end{eqnarray}
and that these expansions can be differentiated term by term.  By inserting these into the scattering Problem \ref{strong}, the results announced in Section \ref{sec:overview} emerge.

\subsection{Expansion of the Maxwell system}

The first of the Maxwell equations, $\gradp\dot E^\eta - \io\,\mu^\eta h^\eta = 0$, gives
\begin{equation*}
  {\eta}^{-1} \gradpy\dot E^0(x,y) + (\gradpx\dot E^0(x,y) + \gradpy\dot E^1(x,y)) - \io\,\mu(x,y) h^0(x,y) = \bigo(\eta),
\end{equation*}
which yields the equations
\begin{equation}\label{expansion1}
  \pair{.}{\gradpy\dot E^0(x,y) = 0,}
           {\gradpx\dot E^0(x,y) + \gradpy\dot E^1(x,y) - \io\,\mu(x,y) h^0(x,y) = 0,}{\}} \quad \text{for $y\notin\dD$}.
\end{equation}
Similarly, the second of the Maxwell equations, $\gradp h^\eta - \io\,\e E = 0$, gives
\begin{equation*}
  {\eta}^{-1} \gradpy h^0(x,y) + (\gradpx h^0(x,y) + \gradpy h^1(x,y)) - \io \e(x,y) E^0(x,y) = \bigo(\eta),
\end{equation*}
which yields the equations
\begin{equation}\label{expansion2}
  \pair{.}{\gradpy h^0(x,y) = 0,}
           {\gradpx h^0(x,y) + \gradpy h^1(x,y) - \io\,\e(x,y) E^0(x,y) = 0,}{\}} \quad \text{for $y\notin\dD$}.
\end{equation}
The interface conditions tell us that
\begin{eqnarray*}
  && E^0\dot t \; \text{ and } \; E^1\dot t \; \text{ are continuous on $\dD$}, \\
  && \pair{.}{E^0\dot t = 0,}
                 {\he^0 - \hi^0 + \sigma E^1\dot t = 0,}{\}} \; \text{on $\dD$}.
\end{eqnarray*}
Let a micro-cell $(\QQ+m)/\eta$ with $m\in\mathbb{Z}^2$ depending on $\eta$ be chosen such that it contains a fixed point $x$, and set $\hat x = m\eta$.  If we integrate in $x$ over the part of $\Oi$ and its boundary contained in this micro-cell, we obtain, after making the change of variable $x'=\hat x + \eta y$,
\begin{equation}
  \eta\int_\dD \frac{\eta}{\sigma} (\he^\eta - \hi^\eta) \d s(y)
     = -\eta \!\int_\dD \!\!E^\eta\cdot t \, \,\d s(y)
  = -\eta^2\!\int_\D \!\! (\gradp\dot E^\eta)  \,\d A(y)
   = -\io\,\eta^2 \!\int_\D \!\! \mu h^\eta  \,\d A(y),
\end{equation}
in which the functions of $x'$ and $y$ are evaluated at $(\hat x + \eta y, y)$.  The expansions of $E^\eta$ and $h^\eta$ then yield the equations
\begin{equation}\label{EMF}
  \int_\dD \sigma^{-1}(x,y) (\he^n(x,y)-\hi^n(x,y))\,\d s(y)
    = -\io\,\!\! \int_\D \! \mu(x,y) \hi^n(x,y)\,\d A(y)
\end{equation}
at each order $n=0,1,2,$ {\it etc}.

\subsection{The cell problem and the homogenized system}
\label{subsec:cell}

From the first of the pair \eqref{expansion2}, $\gradpy h^0(x,y) = 0$, which is valid for each $x\in\Omega$ and for $y\not\in \dD$, we infer exterior and interior values of the magnetic field that are independent of $y$:
\begin{equation}
  h^0(x,y) = \pair{\{}{\he^0(x), \; y\in \Dc,}{\hi^0(x), \; y\in \D.}{.}
\end{equation}
The relation between the exterior and interior values is obtained from \eqref{EMF}, which expresses a balance of electromotive forces on the boundary of a single inclusion,
\begin{equation}
\hat\rho(x) (\he^0(x) - \hi^0(x)) + \io \hat\mu(x) \hi^0(x) = 0,
\end{equation}
in which
\begin{equation} \label{one}
  \hat\mu(x) = \int_\D \mu(x,y) \d A(y), \qquad
  \hat\rho(x) = \int_\dD \sigma(x,y)^{-1} \d s(y),
\end{equation}
and we obtain $\hi(x) = m(x)\he(x)$ and thus an expression for the magnetic field in a cell in terms of its value exterior to the inclusion,
\begin{equation}
h^0(x,y) = M(x,y) \he^0(x),
\end{equation}
in which
\begin{equation} \label{two}
  M(x,y) = \pair{\{}{1, \quad y\in \Dc,}{m(x), \; y \in \D,}{.} \qquad
  m(x) = \frac{\hat\rho(x)}{\hat\rho(x) - \io\hat\mu(x)}.
\end{equation}

We next use the first equation of \eqref{expansion1} and the second of \eqref{expansion2}, together with the continuity of $E\dot t$ on the interfaces to obtain the {\em unit cell problem} 
\begin{eqnarray}
  && \gradpy\dot E^0(x,y) = 0, \label{cell1} \\
  && \gradpx h^0(x,y) + \gradpy h^1(x,y) - \io \e(x,y) E^0(x,y) = 0 \quad \text{off $\dD$}, \label{cell2} \\
  && \Ee^0(x,y)\cdot t = \Ei^0(x,y)\cdot t = 0 \quad \text{on $\dD$}, \label{cell3} \\
  && h^1(x,y) \;\text{ and }\; E^0(x,y) \;\text{ periodic in $y$}.
\end{eqnarray}
This problem determines $h^1$ and $E^0$ as functions of $y$, for each value of $x$.  It is an inhomogeneous periodic-Neumann problem with input $\gradx h^0$, in which the exterior and interior first-order corrections, $\he^1$ and $\hi^1$ are decoupled due to the homogeneous boundary condition on the interface $\dD$ and are determined up to two additive constants in $y$, which are functions of $x$.  One relation between these functions of $x$ is provided again by \eqref{EMF}:
\begin{equation} \label{EMF2}
  \int_\dD \sigma^{-1}(x,y) (\he^1(x,y)-\hi^1(x,y))\,\d s(y)
    = -\io\,\!\! \int_\D \! \mu(x,y) \hi^1(x,y)\,\d A(y).
\end{equation}
As we have mentioned in Section \ref{sec:overview}, the interior value of $E^0$ vanishes identically.

To obtain a PDE governing the average exterior $H$ field $\he^0$ and the average $E$ field $\Eav$, we integrate the second equation of the pair \eqref{expansion1} over the unit cell $\QQ$, as well as the second equation of \eqref{expansion2} after applying $\e(x,y)^{-1}$.  The result is the homogenized system
\begin{equation}\label{homogenized}
  \pair{.}{\gradp\dot \Eav(x) - \io\mu^*(x) \he^0(x) \,=\, 0,}
           {\gradp \he^0(x) - \io\e^*(x)\Eav(x) \,=\, 0.}{.}
\end{equation}
The effective magnetic permeability tensor $\mu^*$ is given by
\begin{equation} \label{three}
  \mu^*(x) = \int_\QQ M(x,y) \mu(x,y) \d A(y),
\end{equation}
and the effective dielectric tensor $\e^*$ arises as we have explained in Section \ref{sec:overview}:
%
\begin{multline}\label{Eav}
\Eav(x) = \int_\QQ E^0(x,y) \,\d A(y)  \\
   = \frac{1}{\io} \int_\Dc \e(x,y)^{-1} \left( \gradp \he^0(x) + \gradpy \he^1(x,y) \right) \d A(y)
   = \frac{1}{\io} \e^*(x)^{-1} \gradp \he^0(x).
\end{multline}

In order to justify calling $\e^*$ an effective dielectric permittivity, we must ensure that $\e^* E$ is correctly interpreted as an appropriate cell average of the electric displacement field $D$.  We address this issue in section \ref{discussion}.

\subsection{Corrector functions}

Because of the condition \eqref{cell3} that $\Ee^0(x,y)\cdot t = \Ei^0(x,y)\cdot t = 0$ on $\dD$, the interior and exterior gradients of the corrector function $h^1(x,y)$ are decoupled and are determined as linear functions of $\grad\he^0(x)$ and $\grad\hi^0(x)$ through the cell problem restricted to $\Dc$ and $\D$,
\begin{eqnarray}
 && \gradpy\he^1(x,y) = \Pe(x,y) \gradp\he^0(x) \quad \text{in $\Dc$},\\
 && \gradpy\hi^1(x,y) = \Pi(x,y) \gradp\hi^0(x) \quad \text{in $\D$}.
\end{eqnarray}
The corrector matrices $\Pe$ (resp. $\Pi$) is defined by setting $\Pe\xxi$ (resp. $\Pi\xxi$) equal to the unique solution of the cell problem in which $\xxi$ stands for $\gradp\he^0(x)$ (resp. $\grad\hi^0(x)$).  Posed in their weak forms, these problems are
\begin{eqnarray}
&& \int_\Dc \e(x,y)^{-1}  \label{Pe}
    \left[ \xxi + \Pe(x,y) \xxi \right] \dot \gradp v(y) \d A(y) = 0
\quad \text{for all } v\in \Honeper(\Dc), \\
&& \int_\D \e(x,y)^{-1}    \label{Pi}
    \left[ \xxi + \Pi(x,y) \xxi \right] \dot \gradp v(y) \d A(y) = 0
\quad \text{for all } v\in H^1(\D).
\end{eqnarray}
As we have seen, $\Pi$ admits a very simple form:
\begin{equation}
  \Pi(x,y) \xi  =  -\xi.
\end{equation}
The tensor $\e^*(x)$ is defined through
\begin{equation}\label{alpha}
\e^*(x)^{-1}\xxi = \int_\Dc \e(x,y)^{-1}
    \left[ \xxi + P(x,y) \xxi \right] \d A(y).
\end{equation}

The term of inhomogeneity in the cell problem, $\gradpx h^0(x,y)$, might as well be parsed in terms of the average exterior $H$ field $\he^0$ and its gradient,
\begin{equation}
\gradpx h^0(x,y) = \gradpx \left( M(x,y)\he^0(x) \right) =
  \pair{\{}{\gradp \he^0(x), \hspace{10em} y\in \D,}
               {m(x)\gradp \he^0(x) + \gradp m(x) \he^0(x), \;\; y\in \Dc.}{.}
\end{equation}
From this point of view, the corrector function $h^1(x,y)$ is a linear function of $\gradp \he^0(x)$ and $\he^0(x)$ for any fixed value of $x$, so that it is determined in two parts, given by a corrector matrix $\tilde P(x,y)$ applied to the gradient $\gradp \he^0(x)$ and a corrector vector $\tilde Q(x,y)$, which multiplies the scalar $\he^0(x)$:
\begin{equation}
\gradpy h^1(x,y) = \tilde P(x,y) \gradp \he^0(x) + \tilde Q(x,y) \he^0(x),
\end{equation}
and one computes that these are related to $\Pe$ and $\Pi$ by
\begin{eqnarray}
  && \tilde P(x,y) = \chie(y) \Pe(x,y) - \chii(y) m(x), \\
  && \tilde Q(x,y) = \chii(y) \gradp m(x).
\end{eqnarray}

\subsection{Discussion of average fields}\label{discussion}

The homogenized system \eqref{homogenized} is clearly of Maxwell form, in which $\Eav(x)$ and $\he^0(x)$ play the role of electric and magnetic fields in a bulk two-dimensional medium.  These fields are certain micro-cell averages of the actual electromagnetic fields in the micro-scale composite, $\Eav(x)$ being the usual cell average and $\he^0(x)$ being the average over the part of the cell exterior to the micro-resonator.  We expect additionally that the fields $\e^*(x)\Eav(x)$ and $\mu^*(x)\he^0(x)$ should represent suitable cell averages of the $D$ field $\e E$ and the $B$ field $\langle 0,0,\mu h \rangle$.

In the medium exterior to the slab (air, for example) all of these fields reduce to $E$, $h$, $D=\e E$, and $b=\mu h$.  In the composite, one has be careful about how cell averages are to be understood.  The standard point of view in the physics literature (see, for example, \cite{PendryHoldenRobbins1999} or \cite{Ramakrish2005}) is based on the fact
one should keep in mind the different geometrical roles of these fields:  Averages of the $E$- and $H$ fields should be taken over line segments traversing a micro-period in the direction of each field component separately because these fields are naturally integrated over one-dimensional curves (they are fundamentally one-forms).  On the other hand, averages of the $D$- and $B$ fields should be taken over sides of the unit cube perpendicular to each component separately because these fields are naturally integrated over surfaces (they are fundamentally two-forms).

In the two-dimensional $H$-polarized reduction, this scheme amounts to computing averages in the following way.  Since $H$ is out of plane, its average is taken simply by evaluating $h$ at a suitable point, whereas the average of $b$ must be taken over the entire unit cell.  The $x_i$-component of the $E$ field should be averaged over line segment in the $x_i$-direction, whereas the average of the $x_i$-component of the $D$ field reduces to the average over a line segment in the $x_{i'}$-direction, where $i' = (i+1)(\mathrm{mod}2)$, because $D$ is constant in the out-of-plane direction.

We now demonstrate that, from this point of view, the tensors $\e^*(x)$ and $\mu^*(x)$ that emerge in the homogenized equations do indeed relate the natural cell averages of $E$ and $D$ to each other as well as those of $h$ and $b$.

In defining $\hav$, a choice between its evaluation in the exterior or the interior of the resonator needs to be made.  In problems of scattering by a slab, one avoids cutting through a micro-resonator and therefore arranges the slab-air interface such that the air is incident with the dielectric matrix of the composite that is exterior to the resonators.  The exterior $H$ field is therefore the one that connects continuously with the $H$ field in the air, and therefore we should work with the average of $h$ in the exterior:
\begin{equation}
  \hav(x) := \he^0(x).
\end{equation}
Evidently, since $h$ is constant in $y$ to leading order exterior to the resonator and constant in the out-of-plane direction, no averaging of oscillations is truly taking place in this definition.

The $B$ field, which is identified with its out-of-plane component $b \!=\! \mu h$, should be averaged over the cell $\QQ$:
\begin{equation}
  \bav(x) := \int_\QQ \mu(x,y) h^0(x,y) \,\d A(y),
\end{equation}
which is exactly equal to $(\io)^{-1}$ times the second term in the right-hand side of the first equation in \eqref{homogenized}, and we have therefore
\begin{equation}
  \bav(x) = \mu^*(x) \hav(x).
\end{equation}

Now let us examine the $E$ and $D$ fields.
Let $e_1$ and $e_2$ denote the vectors $\langle 1,0 \rangle$ and $\langle 0,1 \rangle$, and set $i' = (i+1)(\mathrm{mod}2)$.  The cell average of the $e_i$ component of the $E$ field should be taken over a line segment in the direction of $e_i$ traversing a period cell,
\begin{equation}
\Eav(x) \dot e_i := \int\limits_{\makebox[2.2em][b]{\parbox{2.2em}{\scriptsize $y_i=0$ \\ $y_{i'}=c$}}}^1
E^0(x,y)\dot e_i \,\d y_i.
\end{equation}
These integrals are independent of $y_{i'}$ because $\gradpy\dot E^0(x,y) = 0$, and therefore by this definition, $\Eav$ can be taken to be the area integral of $E^0(x,y)$ over $\QQ$, as we have defined it above in \eqref{Eav}.

The $e_i$ component of the $D$ field $D=\e E$ is averaged over a line segment in the direction of $e_{i'}$,
\begin{equation}
  \Dav(x) \dot e_i := \int\limits_{\makebox[2.2em][b]{\parbox{2.2em}{\scriptsize $y_{i'}=0$ \\ $y_{i}=c$}}}^1
  \e(x,y) E^0(x,y) \dot e_i \,\d y_{i'}.
\end{equation}
If the path of integration remains exterior to $\D$, this yields
\begin{equation}
  \Dav(x) \dot e_i = \frac{1}{\io}
  \int\limits_{\makebox[2.2em][b]{\parbox{2.2em}{\scriptsize $y_{i'}=0$ \\ $y_{i}=c$}}}^1
  \left( \gradpx \he^0(x) + \gradpy h^1(x,y) \right) \dot e_i \,\d y_{i'}
  = \frac{1}{\io} \gradpx \he^0(x)\dot e_i \,,
\end{equation}
and therefore, with this definition of $\Dav$, we obtain from the homogenized equation \eqref{homogenized},
\begin{equation}
  \Dav(x) = \e^*(x) \Eav(x).
\end{equation}

\section{Two-scale convergence analysis} 
\label{sec:twoscale}

This Section establishes with mathematical rigor the results of the formal analysis.  Our main result is that the solution of the problem of scattering by the micro-structured slab tends to the solution of the problem of scattering by a homogeneous slab.  This amounts to convergence of the electromagnetic fields to the average fields discussed in the previous section and the fact that these average fields satisfy the effective system \eqref{homogenized} in the slab.  Of course we must also show that the field $\hav = \he^0$ connects continuously to the function $h$ in the air and that the tangential component of $\Eav$ at the edge of the slab connects continuously to that of $E$ from the side of the air.

These things are properly handled by means of the weak form of the scattering problem.  The outgoing conditions for the scattered field are most conveniently expressed by employing the Dirichlet-to-Neumann map $T$ on the sides $\Gamma_\pm$ for outgoing fields.  This means that fields of the form \eqref{outgoing} are characterized equivalently by
\begin{equation}
\frac{\partial h_\mathrm{sc}}{\partial n}|_{\Gamma_\pm} = 
   -T \left( h_\mathrm{sc}|_{\Gamma_\pm} \right).
\end{equation}
For $g\in H^{\half}(\Gamma_\pm)$, and $\{\hat g^\kappa_m\}$ denoting the Fourier transform of $ge^{-\ik x_2}$, this map is defined through
\begin{equation}
(\widehat{Tg})^\kappa_m = -\ii\nu_m \hat g^\kappa_m\,.
\end{equation}
$T$ is a bounded operator from $H^{\half}(\Gamma_\pm)$ to $H^{-\half}(\Gamma_\pm)$, with norm $\| T \|_{H^{1/2}\to H^{-1/2}} = C_1 + C_2|\omega|$ for some positive constants $C_1$ and $C_2$.  For real values of $\omega$ and $\kappa$, it has a finite-dimensional negative imaginary subspace spanned by the functions
$e^{\ii (m+\kappa) x_2}$ for which $\nu_m^2 > 0$ and a positive subspace equal to the space orthogonal to the imaginary one.  We denote the corresponding decomposition of $T$ by $T = T^+ + \ii T^-$.

Denote by $\Honek(\Oeeta)$ the $\kappa$-pseudoperiodic subspace of $H^1(\Oeeta)$:
\begin{equation}
\Honek(\Oeeta) = \left\{ u\in H^1(\Oeeta) : u|_{x_2=2\pi} = e^{2\pi \ii\kappa} u|_{x_2=0} \right\}.
\end{equation}
It is natural to identify $\Honek(\Oeeta)$ with a subspace of $L^2(\Omega)$ by extension by zero into $\Oieta$, and $H^1(\Oieta)$ with a subspace of $L^2(\Omega)$ by extension by zero into $\Oeeta$.  We can then define the space $V^\eta$ in which $h^\eta$ resides by
\begin{equation}
\Honek(\Oeeta)\oplus H^1(\Oieta) =: V^\eta \hookrightarrow L^2(\Omega) \hookrightarrow (V^\eta)^* \cong V^\eta.
\end{equation}
It will be convenient also to define $V^0 := \Honek(\Omega)$.

The following weak form of the scattering problem is equivalent to the strong form for smooth $\e$, $\mu$, and $\sigma$ and for functions $\he$ and $\hi$ that are smooth in $\overline\Oe$ and $\overline\Oi$\,.

\begin{problem}[Scattering by a slab, weak form] 
\label{microweak}
Find a function $h^\eta_\omega = h^\eta = \he^\eta\oplus\hi^\eta \in V^\eta$ such that
\begin{multline}
\int_\Omega \left[ \left( (\e^\eta)^{-1}\gradp h^\eta \right)\dot\overline{\gradp v} - \omega^2\mu^\eta\, h^\eta \bar v \right] \d A
- \io \int_\dOieta \! \eta (\sigma^\eta)^{-1}(\he^\eta-\hi^\eta)(\bar\ve-\bar\vi) \,\d s \,+ \\
+  \e_0^{-1} \int_{\Gamma_\pm} \!(T \he^\eta)\bar\ve \,\d x_2
= - 2 \ii\nu_{\bar m}\e_0^{-1} \int_{\Gamma_-} \!\! e^{\ii((\bar m+\kappa)x_2 + \nu_{\bar m} x_1)} \bar\ve \,\d x_2 
\end{multline}
for all $v = \ve\oplus\vi \in V^\eta$, and let $E^\eta\in L^2(\Omega)$ be defined by
\begin{equation}
E^\eta = \frac{1}{\io}\,(\e^\eta)^{-1}\gradp h^\eta \quad \text{in $\Oeeta\cup\Oieta$}.
\end{equation}
\end{problem}

This scattering problem always has a solution.  It is possible to prove that on finite intervals of the real $\omega$-axis avoiding  a discrete set of frequencies, $h^\eta_\omega$ is unique for all $\eta$ sufficiently small, but this fact will not needed in our analysis.

\begin{theorem}[Existence of solutions] 
\label{existence}
For each frequency $\omega$ and $\eta\in\Upsilon$, the scattering Problem \ref{microweak} has a solution $h^\eta$.
\end{theorem}

\noindent
Before proving the Theorem, we look more carefully at the forms involved in the weak formulation of the scattering problem.

Let the form $\aeo(u,v)$ in $V^\eta$ be defined by the left-hand side of the equality in Problem \ref{microweak}, with $u$ in place of $h^\eta$.  We split $\aeo$ into two parts:
\begin{eqnarray}
  \aeo(u,v) &=& \beo(u,v) - \omega^2 \ce(u,v)\,, \\
  \beo(u,v) &:=& \int_\Omega
          \left( (\e^\eta)^{-1}\gradp u \right)\dot\overline{\gradp v} \,\d A
  \,+ \\  &&  \notag
  - \,\io\!\! \int_\dOieta \! \eta (\sigma^\eta)^{-1}(\ue-\ui)(\bar\ve-\bar\vi) \,\d s
+  \e_0^{-1} \int_{\Gamma_\pm} \!(T \ue)\bar\ve \,\d x_2\,, \\
  \ce(u,v) &:=& \int_\Omega \mu^\eta\, u \bar v \,\d A\,.
\end{eqnarray}
The incident plane-wave field provides the forcing function for the scattering problem, given by the
functional $f$ on $V^\eta$ defined by
\begin{equation}\label{f}
f(v) := -2\ii\nu_{\bar m} \e_0^{-1}\int_{\Gamma_-} \!\! e^{\ii((\bar m+\kappa)x_2 + \nu_{\bar m} x_1)} \bar\ve \,\d x_2\,.
\end{equation}
The form $\beo$ is coercive, with constants that are independent of $\eta$, and $f$ is uniformly bounded:
\begin{eqnarray}
  && |\beo(u,v)| \leq (\gamma_1 + \gamma_2|\omega|) \|u\|_{V^\eta} \|v\|_{V^\eta}\,, \\
  && \Re b^\eta(u,u) \geq \delta \|u\|^2_{V^\eta}\,, \label{coercivity} \\
  && \| f \|_{(V^\eta)^*} < C \,.
\end{eqnarray}
The second inequality holds because of the upper bound on $\e^\eta$.  In order to establish the upper bound on $|a^\eta(u,v)|$, one uses the bound on $T$ discussed above and the following Lemma~\ref{interfaceterm} showing that the term involving the interface $\dOieta$ is a bounded form on $V^\eta$.  The bound on $f$ is proved in Lemma \ref{forcing}.

\begin{lemma} 
\label{interfaceterm}
\label{extension} 
There exists a constant $C$ such that the following hold for each $\eta~\in~\Upsilon$.
\begin{enumerate}
  \item For all $u = \ue\oplus\ui$ and $v = \ve\oplus\vi$ in $V^\eta$, the following inequality holds:
\begin{equation}
\left| \int_{\dOieta} \eta(\sigmas^\eta)^{-1} (\ue-\ui)(\bar\ve-\bar\vi) \d s \right| \leq C \|u\|_{V^\eta} \|v\|_{V^\eta}.
\end{equation}
  \item For each $u = \ue\oplus\ui \in V^\eta$, there are functions $\tilde\ue\in\Honek(\Omega)$ and $\tilde\ui\in H^1(\Omega)$ such that $\tilde\ue|_{\Oieta} = \ue|_{\Oieta}$ and $\tilde\ui|_{\Oieta} = \ui|_{\Oieta}$ and
\begin{equation}
  \| \tilde\ue \|_{H^1(\Omega)} \leq C \| u \|_{H^1(\Oeeta)}, \qquad
  \| \tilde\ui \|_{H^1(\Omega)} \leq C \| u \|_{H^1(\Oieta)}.
\end{equation}
\end{enumerate}
\end{lemma}

\begin{proof}
To prove part (1),
we use a Poincar\'e inequality in the square $\QQ$: For $\we\in H^1(\Dc)$,
\begin{equation}\label{Poincare1}
\int_\dD |\we(y)|^2 \d s(y) \,\leq\, C\!\!\int_{\Dc} \!\left( |\we(y)|^2 + |\grad\we(y)|^2 \right) \d A(y).
\end{equation}
Rescaling this estimate to a cell of size $\eta$ in the variable $x$ and summing over all cells in $\Omega_0$ yields
\begin{equation}
\int_{\dOieta} |\ue(x)|^2 \d s(x)
  = \frac{C}{\eta} \int_\Oeeta \left( |\ue(x)|^2 + \eta^2|\grad\ue(x)|^2 \right) \d A(x).
\end{equation}
This, together with an analogous estimate involving $\Oieta$ give
\begin{equation}
\int_{\dOieta} |\ue|^2 \d s  \leq  \frac{C}{\eta} \| \ue \|^2_{H^1(\Oeeta)}, \quad
\int_{\dOieta} |\ui|^2 \d s  \leq  \frac{C}{\eta} \| \ui \|^2_{H^1(\Oieta)}.
\end{equation}
Finally, recalling that $\sigma_{\!-}\leq \Re\sigma(x,y)$, we obtain
\begin{multline}
\left| \int_{\dOieta} \eta(\sigmas^\eta)^{-1} (\ue-\ui)(\bar\ve-\bar\vi) \d s \right|^2
\leq \eta^2 (\sigma_{\!-})^{-2} \int_{\dOieta} |\ue-\ui|^2 \d s  \int_\dOieta |\ve-\vi|^2 \d s \\
\leq C^2(\sigma_{\!-})^{-2} \left( \|\ue\|_{H^1(\Oeeta)} + \|\ui\|_{H^1(\Oieta)} \right)^2 \!
                                        \left( \|\ve\|_{H^1(\Oeeta)} + \|\vi\|_{H^1(\Oieta)} \right)^2
\leq C^2(\sigma_{\!-})^{-2} \|u\|^2_{V^\eta} \|v\|^2_{V^\eta}.
\end{multline}

An extension of $u=\ue$ in part (2) can be obtained from the extension operator defined in \cite{CioranescSaint-Jea1979}.  For the convenience of the reader and coherence of this text, we include a proof.  The extension is accomplished by defining $\tue$ to be a harmonic function in $\Oieta$ whose trace on $\dOi$ coincides with that of $\ue$.  The extension is $\kappa$-pseudoperiodic because $\ue$ is and $\dO\cap\dOi=\emptyset$.  Let $\hat x_n = n\eta \in \Omega_0$ where $n\in\mathbb{Z}^2$.
Using the following standard inequalities in the unit cube,
\begin{eqnarray}
&& \int_\QQ  |\grady \tue(\hat x + \eta y)|^2 \,\d A(y) \,\leq\, c \int_{\Dc} |\grady \ue(\hat x+\eta y)|^2 \,\d A(y), \\
&& \int_\QQ |\ue(\hat x + \eta y)|^2 \,\d A(y) \,\leq\, c \int_{\Dc} ( |\ue(\hat x + \eta y)|^2 + |\grady \ue(\hat x + \eta y)|^2 ) \,\d A(y),
\end{eqnarray}
we obtain the estimates in the scaled cubes $\hat x_n + \eta \Dc$\,,
\begin{equation}
  \int_{\eta\QQ} |\gradx \tue(\hat x+ x)|^2 \,\d A(x) 
    \leq c\int_{\eta \Dc} |\gradx \ue(\hat x + x)|^2 \,\d A(x)
\end{equation}
and
\begin{equation}
  \int_{\eta\QQ} |\tue(\hat x + x)|^2 \,\d A(x) 
    \leq 
  = c\!\! \int_{\eta \Dc} \!\!(|\ue(\hat x + x)|^2 + \eta^2 |\gradx \ue(\hat x + x)|^2) \,\d A(x).
\end{equation}
Summing up over all micro-cells in $\Omega_0$, we obtain
\begin{eqnarray}
 && \int_{\Omega_0} |\gradx \ue|^2 \,\d A(x) \,\leq\, c \int_{\Omega_{0\mathrm{e}}} |\gradx \ue|^2 \,\d A(x), \\
 && \int_{\Omega_0} |\ue|^2 \,\d A(x) \,\leq\, c \int_{\Omega_{0\mathrm{e}}} (|\ue|^2 + \eta^2|\gradx \ue|^2) \,\d A(x),
\end{eqnarray}
The result trivially extends to $\Omega$ because the edges of the slab $\Sigma_\pm$ do not intersect the resonators.  The extension of $\ui$ is handled identically in $\Omega_0$ and then extended to a function in $H^1(\Omega)$ by a harmonic function in $\Omega\!\setminus\!\Omega_0$ whose trace coincides with $\tilde\ui$ on $\Sigma_\pm$.  The extension to $\Omega$ can be taken to be in $\Honek(\Omega)$ if $\tilde\ui\in\Honek(\Omega_0)$.
\end{proof}
\medskip


\renewcommand{\labelenumi}{\alph{enumi}.}

\begin{lemma} 
\label{forcing} 
Let $f\in (V^\eta)^*$ be defined as in \eqref{f}.  There exists a constant $C$ such that $\| f \|_{(V^\eta)^*} < C$ for all $\eta\in\Upsilon$.
\end{lemma}

\begin{proof}

Let $u = \ue \oplus \ui \in V^\eta$ be given, and let $\tue$ be the extension of $\ue$ provided by part (2) of Lemma \ref{extension}.  Then
\begin{equation}
|f(u)| = |f(\ue)| = |f(\tue)| \leq \|f\|_{H^1(\Omega)^*} \|\tue\|_{H^1(\Omega)}
\leq c\|f\|_{H^1(\Omega)^*} \|\ue\|_{H^1(\Oeeta)} \leq c\|f\|_{H^1(\Omega)^*} \|u\|_{V^\eta},
\end{equation}
which implies
\begin{equation}
\|f\|_{(V^\eta)^*} \leq c \| f \|_{H^1(\Omega)^*}.
\end{equation}
\end{proof}

The forms we have defined give rise to bounded operators $\Beo$ and $\Ce$ from $V^\eta$ into itself such that
\begin{eqnarray}
  && \beo(u,v) = (\Beo u, v)_{V^\eta} \,,\\
  && \ce(u,v) = (\Ce u, v)_{V^\eta} \,,
\end{eqnarray}
where $(\cdot,\cdot)_{V^\eta}$ is the usual inner product in $V^\eta$.
Let $\tilde f$ denote that element of $V^\eta$ such that $(\tilde f,\cdot)_{V^\eta} = f$.
The scattering Problem \ref{microweak} now takes the form
\begin{equation} \label{scattering}
  \aeo(u,v) = f(v) \text{ for all $v\in V^\eta$,} \quad \text{or} \quad \Beo u - \omega^2 \Ce u = \tilde f.
\end{equation}

We now prove that the scattering Problem \ref{microweak} has a solution.  We refer to Bonnet-Ben Dhia and Starling \cite{Bonnet-BeStarling1994}, Theorem 3.1, where this is shown for a similar problem of scattering by a slab.

\medskip
\begin{proof}[Proof of Theorem \ref{existence}]
Since $\beo$ is coercive, $\Beo$ is a bijection with bounded inverse, and since $\mu\in L^\infty(\Omega)$, $\Ce$ is compact.  Therefore a Fredholm alternative is applicable: The scattering problem \eqref{scattering} has a solution $u$ if and only if $(\tilde f,v)=0$ for all
$v\in\mathrm{Null}(\Beo - \omega^2 \Ce)^\dagger$, or equivalently
\begin{equation}\label{solvability}
  f(v) = 0 \text{ for all $v\in V^\eta$ such that } \aeo(w,v) = 0 \text{ for all $w\in V^\eta$}.
  \quad\text{(solvability condition)}
\end{equation}
Now any function $v$ satisfying the adjoint eigenvalue condition
$\aeo(w,v) = 0$ for all $w\in V^\eta$ satisfies, in particular
\begin{equation}
  \aeo(v,v) = 0,
\end{equation}
and, since the imaginary part of $\aeo$ is nonpositive, we find that
\begin{equation}
  T_-(v) = 0.
\end{equation}
Therefore the trace of $v$ on $\Gamma_-$ possesses only decaying Fourier harmonics, exterior to the slab,
\begin{equation}
v(x) = \sum_{m:\ii\nu_m<0} \gamma^\pm_m e^{\ii((m+\kappa)x_2 + \nu_m |x_1|)}
\quad \text{ exterior to the slab,}
\end{equation}
whereas $f$ is defined through integration against the trace of the incident field, which possesses only a propagating harmonic.  Thus we obtain
\begin{equation}
f(v) = -2\ii\nu_{\bar m} \int_{\Gamma_-} \!\! e^{\ii((\bar m+\kappa)x_2 + \nu_{\bar m} x_1)}
\left( \sum_{m:\ii\nu_m<0} \bar\gamma^-_m e^{\ii(-(m+\kappa)x_2 + \nu_m x_1)} \right)
\d x_2 \,=\, 0 \,,
\end{equation}
{\it i.e.}, the solvability condition \eqref{solvability} is valid.  Thus, by the Fredholm alternative, there exists a solution $h^\eta_\omega$ to 
\begin{equation}
\aeo(h^\eta_\omega,v) = f(v) \quad \text{for all $v\in V^\eta$}.
\end{equation}
\end{proof}

Since we do not have {\it a priori} bounds on the solutions $h^\eta$, we scale for the time being the magnetic and electric fields by a number $m^\eta\in [0,1]$,
\begin{eqnarray}
  && m^\eta = \min\{ 1, \| h^\eta \|^{-1}_{L^2} \}, \\
  && u^\eta(x) = m^\eta h^\eta(x), \\
  && F^\eta(x) = m^\eta E^\eta(x),
\end{eqnarray}
such that the $u^\eta$ are solutions of the scattering problem with a scaled incident wave and are bounded in the $L^2$ norm uniformly in $\eta$:
\begin{eqnarray}
  && \aeo(u^\eta,v) = m^\eta f(v) \quad \text{ for all $v\in V^\eta$}, \label{eqnueta}\\
  && F^\eta = \frac{1}{\io}(\e^\eta)^{-1} \gradp u^\eta, \\
  && \| u^\eta \|_{L^2} \leq 1.
\end{eqnarray}
After proving the two-scale convergence of the $u^\eta$ to a solution of the homogenized equations \eqref{homogenized}, we will conclude that the $h^\eta$ are in fact uniformly bounded in $V^\eta$ (at the end of the proof of Theorem \ref{thm:vartwoscale}).   This result uses the fact that the scattering problem for the homogeneous slab has a unique solution.  It will be shown that this always holds true in the case that $\mu^\eta$ is real, in particular in the typical case of nonmagnetic materials, $\mu^\eta=\mu_0$.

The next Lemma establishes the existence of two-scale limits\footnote{
A sequence $u^\eta(x)$ two-scale converges to $u^0(x,y)$ ($u^\eta \tsc u^0$) in $\Omega$ if $\lim_{\eta\to0} \int_\Omega u^\eta(x)\phi(x,x/\eta) \,\d A(x)
     = \int_\Omega\int_\QQ u^0(x,y) \phi(x,y) \,\d A(y) \,\d A(x)$ for all smooth functions $\phi(x,y)$ in $\overline\Omega\times\mathbb{R}^2$, periodic in $y$.
}
of subsequences of $u^\eta$ and $\grad u^\eta$ whose $y$-dependence reflects the discontinuities at the surfaces of the micro-resonators.  No other feature of the weak-form Problem \ref{microweak} is used here.  The symbol "$\tsc$" means ``two-scale converges to".  We denote the characteristic function of $\D$ extended periodically to $\mathbb{R}^2$ by $\chii(y)$ and set $\chie(y) = 1-\chii(y)$.

\begin{lemma}[Two-scale limits] 
Every subsequence of $\Upsilon$ admits a subsequence $\Upsilon'$ and functions
$\ue^0(x) \in H^1_\kappa(\Omega), \, \ui^0(x) \in H^1_\kappa(\Omega_0)$,
$\ue^1(x,y) \in L^2(\Omega;\Honeper(\QD)/\mathbb{R})$, and $\ui^1(x,y) \in L^2(\Omega_0;\Honeper(\QQ)/\mathbb{R})$ 
such that, in $\Omega$,
\begin{equation}
\pair{.}{u^\eta(x) \tsc \chie(y)\ue^0(x) + \chii(y)\ui^0(x) \vspace{3pt}}
           {\grad u^\eta(x) \tsc \chie(y)[\grad \ue^0(x) + \grady \ue^1(x,y)]
                                                      + \chii(y)[\grad \ui^0(x) + \grady \ui^1(x,y)] \vspace{3pt}}
           {\}},  \quad \eta\in\Upsilon'.
\end{equation}
\end{lemma}

\begin{proof}
With $v=u^\eta$ in \eqref{eqnueta}, we obtain
\begin{equation}
b(u^\eta,u^\eta) = m^\eta f(u^\eta) + \omega^2 c(u^\eta,u^\eta).
\end{equation}
Then, using the coerciveness of $\beo$, the boundedness of $\mu$, and part (b) of Lemma \eqref{extension}, we obtain the estimate 
\begin{multline}
\delta \| u^\eta \|^2_{V^\eta} \leq \Re b(u^\eta,u^\eta)
  \leq |f(u^\eta)| + \omega^2 \int_\Omega \Re \mu^\eta|u^\eta|^2 \,\d A(x) \\
  \leq C\|u\|_{V^\eta} + \omega^2\mu_+ \|u^\eta\|^2_{L^2}
  \leq (C+\omega^2\mu^+) \| u^\eta \|_{V^\eta},
\end{multline}
from which it follows that
\begin{equation}
\| u^\eta \|_{V^\eta} \leq (C+\omega^2\mu_+)/\delta\,.
\end{equation}

We now have that $u^\eta$ and $\grad u^\eta$ (defined for points off $\dOieta$) are bounded sequences in $L^2(\Omega)$, and therefore, by Theorem 1.2 of Allaire \cite{Allaire1992}, there exist functions $u^0\in L^2(\Omega\times\QQ)$
and $\xi^0\in L^2(\Omega\times\QQ)^2$ and a subsequence $\Upsilon'\subset\Upsilon$ such that
\begin{equation}
  u^\eta(x) \tsc u^0(x,y), \;\text{ and }\; \grad u^\eta(x) \tsc \xi^0(x,y), \quad \eta\in\Upsilon'.
\end{equation}
The behavior of these functions is well known to be trivial in the region $\Omega\setminus\Omega_0$ outside the slab, where they have no $y$-dependence, and $u^0$, which we denote by $\ue^0$ in this region, is of class $H^1$ there.

Let
$\Psi\in C^\infty_0(\Omega_0;C^\infty_\#(\mathbb{R}^2))^2$ with $\Psi(x,y)\dot n = 0$ for $y\in\dD$ be given.  Upon integration by parts, the latter condition eliminates integrals over $\dOieta$ which would otherwise arise due to the discontinuity in $u^\eta$ on the interfaces of the resonators:
\begin{equation}\label{parts}
  \int_\Omega \grad u^\eta(x) \dot \Psi(x,x/\eta) \d x
  + \int_\Omega u^\eta(x) \left( \gradx\dot\Psi(x,x/\eta)
        + {\textstyle\frac{1}{\eta}} \grady\dot\Psi(x,x/\eta) \right) \d x \,=\, 0.
\end{equation}
Multiplying this equality by $\eta$ and using the two-scale convergence of $u^\eta$ and $\grad u^\eta$, the limit for $\eta\in\Upsilon'$ gives
\begin{equation}\label{hi}
\int_\Omega \int_\QQ u^0(x,y) \grady\dot\Psi(x,y) \,\d y \,\d x = 0.
\end{equation}
Let $\psi\in C^\infty_0(\Omega_0)$ be given and use a separable function $\phi(x)\Phi(y)$, with $\Phi(y)\dot n=0$ on $\dD$, in place of $\Psi(x,y)$ in \eqref{hi} to obtain
\begin{equation}
\int_\Omega \phi(x) \int_\QQ u^0(x,y) \grad\dot\Phi(y) \,\d y \,\d x = 0,
\end{equation}
from which we obtain
\begin{equation}
  \int_\QQ u^0(x,y) \grady\dot\Phi(y)\,\d y = 0,
\end{equation}
for each such $\Phi(y)$.  It follows from this and the connectedness of $\D$ and $\Dc$, that $u^0(x,y)$ is independent of $y$ in each of these domains.  Therefore, for $x\in\Omega_0$,
\begin{eqnarray}
 && u^0(x,y) = \chie(y) \ue^0(x) + \chii(y) \ui^0(x), \\
 && \ue^\eta(x) \tsc \chie(y) \ue^0(x), \\ 
 && \ui^\eta(x) \tsc \chii(y) \ui^0(x).
\end{eqnarray}
To prove that $\ue^0(x)\in H^1(\Omega_0)$, we must prove that there is a constant such that, for each smooth function $\Theta(x)$ with compact support in $\Omega_0$,
\begin{equation}
\left| \int_{\Omega_0} \ue^0(x) \grad\dot\Theta(x) \d x \right| \leq \text{const.} \| \Theta \|_{L^2(\Omega_0)^2}.
\end{equation}
By Lemma 2.10 of \cite{Allaire1992}, we can further impose upon $\Psi$ (in addition to $\Psi(x,y)\dot n = 0$ on $\dD$) the three properties (i) $\grady\dot\Psi(x,y)\!=\!0$, (ii) $\int_{\Dc} \Psi(x,y) \d y = \Theta(x)$, and
(iii) $\|\Psi\|_{L^2(\Omega_0\times \Dc)^2}
\,\leq\, c \|\Theta\|_{L^2(\Omega_0)^2}$.
One one hand, we obtain
\begin{multline}
\int_{\Omega_0} \grad\ue^\eta(x) \dot \Psi(x,x/\eta) \d x = \int_{\Omega_0} \ue^\eta(x) \gradx\dot\Psi(x,x/\eta) \d x
   \longrightarrow \int_{\Omega_0}\int_{\Dc} \ue^0(x) \gradx\dot\Psi(x,y) \d y \d x \\
                         = \int_{\Omega_0} \ue^0(x) \grad\dot\! \int_{\Dc} \!\!\Psi(x,y) \d y \d x
                         = \int_{\Omega_0} \ue^0(x) \grad\dot \Theta(x) \d x,
   \quad \eta\in\Upsilon'.
\end{multline}
On the other hand,
\begin{equation}
\int_{\Omega_0} \grad\ue^\eta(x) \dot \Psi(x,x/\eta) \d x \longrightarrow
       \int_{\Omega_0}\int_{\Dc}\xi^0(x,y)\dot\Psi(x,y) \d y \d x, \quad \eta\in\Upsilon',
\end{equation}
so that 
\begin{equation}
\int_{\Omega_0} \ue^0(x) \grad\dot \Theta(x) \d x
\,=\, \int_{\Omega_0}\int_{\Dc}\xi^0(x,y)\dot\Psi(x,y) \d y \d x, \quad \eta\in\Upsilon',
\end{equation}
and the inequality
\begin{equation}
\left| \int_{\Omega_0}\int_{\Dc}\xi^0(x,y)\dot\Psi(x,y) \d y \d x \right| 
\,\leq\, \|\xi^0\|_{L^2(\Omega_0\times \Dc)^2} \|\Psi\|_{L^2(\Omega_0\times \Dc)^2}
\,\leq\, c \|\xi^0\|_{L^2(\Omega_0\times \Dc)^2} \|\Theta\|_{L^2(\Omega_0)^2}
\end{equation}
demonstrates that $\grad \ue^0$ acts as a bounded linear functional on $L^2(\Omega_0)^2$ so that $\ue^0\in H^1(\Omega_0)$.  An analogous argument proves that $\ui^0 \in H^1(\Omega_0)$.  One can see that $\ue^0$ and $\ui^0$ are in $\Honek(\Omega_0)$ by applying the arguments above to an extension of $\Omega_0$ into a neighborhood of $x_2=2\pi$ with $u^\eta$ extended pseudo-periodically into this region.

To obtain the functions $\ue^1(x,y)$ and $\ui^1(x,y)$ in $\Omega_0$, we use
$\lim_{\eta\to0}\int_{\Omega_0} \grad\ue^\eta(x) \dot \Psi(x,x/\eta) \d x
= \int_{\Omega_0}\int_{\Dc} \grad \ue^0(x)\cdot\Psi(x,y) \d y \d x$
to obtain $\int_{\Omega_0}\int_{\Dc} \left( \xi^0(x,y) - \grad\ue^0(x) \right) \dot\Psi(x,y) \d y \d x = 0$,
which implies, as before, that
\begin{equation}
\int_{\Dc} \left( \xi^0(x,y) - \grad\ue^0(x) \right) \dot\Phi(y) \d y \d x = 0
\end{equation}
for all $\Phi(y)\in\Honeper(\QQ)^2$ with $\grad\dot\Phi\!=\!0$ and $\Phi\dot n\!=\!0$ on $\dD$,
from which we infer the existence of a function $\ue^1(x,y) \in L^2(\Omega_0;\Honeper(\Dc))$ such that
\begin{equation}
\xi^0(x,y) = \grad\ue^0(x) + \grady\ue^1(x,y), \quad y\in \Dc.
\end{equation}
In an analogous manner, one establishes the existence of $\ui^1(x,y) \in L^2(\Omega_0;\Honeper(\D))$ such that
\begin{equation}
\xi^0(x,y) = \grad\ui^0(x) + \grady\ui^1(x,y), \quad y\in \D.
\end{equation}

To prove that $\ue^0\in H^1(\Omega)$, we must prove that the trace of $\ue^0$ from the right and left of the slab boundaries $\Sigma_\pm$ are equal.  By Lemma \ref{extension}, the functions $\ue^\eta$ can be extended to functions $\tilde u^\eta\in H^1(\Omega)$ in such a way that $\| \tilde u^\eta\|_{H^1(\Omega)} < C$.  We may then extract a subsequence that weakly converges in $H^1(\Omega)$, and by Proposition 1.14 (i) of \cite{Allaire1992}, this subsequence two-scale converges to its weak limit, which must evidently be equal to $\ue^0$.
\end{proof}

The key to obtaining two-scale convergence of the $h^\eta$ to a solution of the scattering problem for a homogeneous slab is the uniqueness of this solution.  The weak form of this problem is given by \eqref{var4}.  If $\Im(\mu^*(x))>0$, then uniqueness is guaranteed.
As noted in section \ref{subsec:model}, equation \eqref{Immu}, $\Im(\mu^*(x))$ is indeed positive provided that $\mu^\eta(x,y)$ is pointwise real and positive.
{\em Theorems \ref{thm:vartwoscale}, \ref{thm:strong}, and \ref{thm:transmission} require as implicit assumptions that $\Im(\mu^*(x))>0$, or, more generally, that \eqref{var4} admits a unique solution.}

\begin{theorem}[Variational problem for the two-scale limit] 
\label{thm:vartwoscale}

The two-scale limiting functions $\he^0(x)$, $\hi^0(x)$, $\he^1(x,y)$, and $\hi^1(x,y)$ satisfy the following variational problem:
\begin{multline}\label{variationaltwoscale}
\int_\Omega \int_{\D*} \left( \e(x,y)^{-1}[\gradp \he^0(x) + \gradpy \he^1(x,y)] \right)
           \dot [\gradp \bar\ve^0(x) + \gradpy \bar\ve^1(x,y)] \d A(y) \d A(x) +  \\
\int_\Omega \int_{\D} \left( \e(x,y)^{-1}[\gradp \hi^0(x) + \gradpy \hi^1(x,y)] \right)
           \dot [\gradp \bar\vi^0(x) + \gradpy \bar\vi^1(x,y)] \d A(y) \d A(x) +  \\
- \omega^2 \int_\Omega \left[ \int_{\Dc} \mu(x,y) \d A(y) \right] \he^0(x) \bar\ve^0(x) \d A(x)
- \omega^2 \int_\Omega \left[ \int_{\D} \mu(x,y) \d A(y) \right] \hi^0(x) \bar\vi^0(x) \d A(x) + \\
- \io \int_\Omega \left[ \int_{\dD} \sigma(x,y)^{-1} \d s(y) \right] (\he^0(x)-\hi^0(x))(\bar\ve^0(x)-\bar\vi^0(x)) \d A(x) \\
= - 2 \ii\nu_{\bar m} \e_0^{-1}\int_{\Gamma_-} \!\! e^{\ii((\bar m+\kappa)x_2 + \nu_{\bar m} x_1)} \bar\ve \,\d x_2                   
\end{multline}
for all $\ve^0(x), \, \vi^0(x) \in \Honek(\Omega)$ and
$\ve^1(x,y) \in L^2(\Omega_0;\Honeper(\QD)) \oplus L^2(\Omega\!\setminus\!\Omega_0)$, $\vi^1(x,y) \in L^2(\Omega_0;\Honeper(\D))$.

This problem is equivalent to the the system
\begin{eqnarray}
  && \hi^0(x) = m(x) \he^0(x), \label{var1}\\
  & &\gradpy h^1(x,y) = \chie(y)\Pe(x,y)\gradp\he^0(x) +  \chii(y)\Pi(x,y)\gradp\hi^0(x), \label{var2}\\
  &&  E^0(x,y) = \frac{1}{\io} \e(x,y)^{-1}
         \left[ \grad h^0(x) + \gradpy h^1(x,y) \right], \label{var3}\\
  && \pair{\{}{\Eav(x) = {\displaystyle\frac{1}{\io}} \e^*(x)^{-1} \gradp \he^0(x),\vspace{1ex}}
        {\displaystyle\int_\Omega \left[ \io \Eav(x) \dot \gradp\bar v(x) - \omega^2 \mu^*(x) \he^0(x)\bar v(x) \right] \d A(x)
   \\ \hspace{9em}
   = \displaystyle - 2 \ii\nu_{\bar m} \e_0^{-1} \int_{\Gamma_-} \!\! e^{\ii((\bar m+\kappa)x_2 + \nu_{\bar m} x_1)} \bar v \,\d x_2
        \quad \text{for all } v \in H^1(\Omega),}{.} \label{var4}
\end{eqnarray}
in which the average $E$ field $\Eav$ is defined as
\begin{equation}\label{Eav2}
\Eav(x) := \int_\QQ E^0(x,y) \d A(y).
\end{equation}
\end{theorem}

\begin{proof}   
Let $\Upsilon'$ be a subsequence of $\Upsilon$ such that $m^\eta$ converges, say to $m^0\in [0,1]$.  We use test functions of the form
\begin{equation}
v(x) = \chie(x/\eta)\left[ \ve^0(x) + \eta \ve^1(x,x/\eta) \right] + \chii(x/\eta)\left[ \vi^0(x) + \eta \vi^1(x,x/\eta) \right] = v^0(x) + \eta v^1(x,x/\eta).
\end{equation}
that are smooth off the boundary of the resonators
in the weak-form Problem \ref{microweak}:
\begin{multline}\label{testing}
  \int_\Omega \left[ \e(x,x/\eta)^{-1} \gradp u^\eta(x) \dot
       \left( \gradp \bar v^0(x) + \gradpy \bar v^1(x,x/\eta) + \eta\gradpx \bar v^1(x,x/\eta) \right) \right] \d A \,+ \\
  -\omega^2 \int_\Omega \mu(x,x/\eta) u^\eta(x) \left( \bar v^0(x) + \eta \bar v^1(x,x/\eta) \right) \d A \,+ \\
  -\io \int_\dOieta \eta\sigma^{-1}(x,x/\eta) \left( \ue^\eta(x) - \ui^\eta(x) \right)
                   \left( \bar\ve^0(x) - \bar\vi^0(x) + \eta( \bar\ve^1(x,x/\eta) - \bar\vi^1(x,x/\eta) )\right) \d s \,= \\
  = - 2 \ii\nu_{\bar m} \int_{\Gamma_-} \!\! e^{\ii((\bar m+\kappa)x_2 + \nu_{\bar m} x_1)}
                      \left( \bar\ve^0(x) + \eta\bar{\ve}^{1}(x,x/\eta) \right) \d x_2\,.
\end{multline}

The following analysis uses two-scale convergence results of \cite{Allaire1992}; we refer in particular to the proof of Theorem 2.3 of that work.
Consider first the first term of \eqref{testing}.  By assumption, $\chi_{\mathrm{e,i}}(y) \e(x,y)^{-1} \in C[\Omega_0;L^\infty_\#(\mathbb{R})^2]^4$, and therefore we also have
\begin{equation}
\triple{.}
  {\zeta_\mathrm{e}(x,y) :=
  \chie(y)(\e(x,y)^{-1})^\dagger \left( \gradp \bar\ve^0(x) + \gradpy \bar\ve^1(x,x/\eta) \right)
  \in C[\Omega_0;L^\infty_\#(\mathbb{R})^2]^2,}
  {\zeta_\mathrm{i}(x,y) :=
  \chii(y)(\e(x,y)^{-1})^\dagger \left( \gradp \bar\vi^0(x) + \gradpy \bar\vi^1(x,x/\eta) \right)
  \in C[\Omega_0;L^\infty_\#(\mathbb{R})^2]^2,}
  {\xi(x,y) = (\e(x,y)^{-1})^\dagger \left( \gradp \bar\ve^0(x) + \gradpy \bar\ve^1(x,x/\eta) \right)
  \in C[\Omega\!\setminus\!\Omega_0;L^\infty_\#(\mathbb{R})^2]^2.}
  {.}
\end{equation}
This implies that the fields $\zeta_\mathrm{e,i}$ and $\xi$ {\em strongly} two-scale converge in their respective domains.\footnote
{\label{footnoteStrong}Strong two-scale convergence of $w^\eta(x)$ to $w^0(x,y)$ means that $w^\eta$ two-scale converges to $w^0$ and that $\int |w^\eta(x)|^2 \d x$ converges to
$\iint |w^0(x,y)|^2 \d x \d y$.}
As $\gradp u^\eta$ two-scale converges in $\Omega_0$, Theorem 1.8 of \cite{Allaire1992}\footnote{\label{footnoteThm18}Essentially, strong two-scale convergence of $u^\eta(x)$ to $u^0(x,y)$ and (weak) two-scale convergence of $v^\eta(x)$ to $v^0(x,y)$ implies weak convergence of $u^\eta(x)v^\eta(x)$ to $\int_\QQ u^0(x,y)v^0(x,y) \d A(y)$.}
justifies passing to the two-scale limit in \eqref{testing} in the entire domain $\Omega$.
A similar argument applies to the second term in \eqref{testing}.

One deals with the limit of the interface term by applying the theory of two-scale convergence on periodic interfaces.  We refer the reader to \cite{AllaireDamlamianHornung1995} and \cite{Neuss-Rad1996} for results on this subject.  Proposition 2.6 of \cite{AllaireDamlamianHornung1995}
applied to the extensions of $\ue^\eta$ and $\ui^\eta$ from Lemma \ref{extension} guarantee the two-scale convergence of the traces of these functions to $\ue^0(x)$ and $\ui^0(x)$.\footnote{
A sequence $w^\eta(x)\in L^2(\dOieta)$ with $\eta\int_\dOieta |w^\eta(x)|^2 \d s(x) < C$ two-scale converges to $w^0(x,y)\in L^2(\Omega\times\dD)$ if \,
$\lim_{\eta\to0} \eta \int_\dOieta w^\eta(x) \phi(x,x/\eta)\,\d s(x) =
\int_\Omega\int_\dD w^0(x,y) \phi(x,y) \,\d s(y) \,\d A(x)$ \, for all continuous $\phi(x,y)$ defined on $\bar\Omega\times\dD$.}
If $\sigma^0(x,y)$ is taken to be continuous, then by extending it to a continuous function on
$\bar\Omega\times\dD$ and incorporating it into the test function, we see that $\sigma^{-1}(x,x/\eta) (\ue^\eta(x)-\ui^\eta(x))$, defined on $\dOieta$, two-scale converges to
$\sigma^{-1}(x,y) (\ue^0(x)-\ui^0(x))$.

The limit for $\eta\in\Upsilon'$ yields the two-scale weak-form PDE system for the two-scale limits $\ue^0$, $\ui^0$, $\ue^1$, and $\ui^1$:
\begin{multline}\label{weaktwoscale}
  \int_\Omega\int_{\Dc} \left[ \e(x,y)^{-1} \left( \gradp\ue^0(x) + \gradpy\ue^1(x,y) \right) \cdot
                      \left( \gradp\bar\ve^0(x) + \gradpy\bar\ve^1(x,y) \right) \right] \d y\, \d x \,+\\
  + \int_\Omega\int_\D \left[ \e(x,y)^{-1} \left( \gradp\ui^0(x) + \gradpy\ui^1(x,y) \right) \cdot
                      \left( \gradp\bar\vi^0(x) + \gradpy\bar\vi^1(x,y) \right) \right] \d y\, \d x \,+\\
  -\omega^2 \int_\Omega \ue^0(x)\bar\ve^0(x) \int_{\Dc}\mu(x,y)\d y\,\d x
  -\omega^2 \int_\Omega \ui^0(x)\bar\vi^0(x) \int_{\D}\mu(x,y)\d y\,\d x \,+ \\
  - \io \int_\Omega \left( \ue^0(x) - \ui^0(x) \right) \left( \bar\ve^0(x) - \bar\vi^0(x) \right)
                                        \int_{\dD} \sigma^{-1}(x,y)\d y\,\d x \\
  \,=\, - 2 \ii\nu_{\bar m} \int_{\Gamma_-} \!\! e^{\ii((\bar m+\kappa)x_2 + \nu_{\bar m} x_1)}
                      \bar\ve^0(x) \d x_2\,, \\
  \text{for all} \; \bar\ve^0,\,\bar\vi^0 \in H^1(\Omega),\;
               \ve^1 \in L^2(\Omega;\Honeper(\Dc)), \, \vi^1 \in L^2(\Omega;\Honeper(\D)).
\end{multline}

Next we pull the two-scale variational problem apart to reveal the equations for the average fields and the cell problem for the corrector function.  In $\Omega_0$, we first put $v^0 = 0$ and obtain
\begin{multline}
  \int_\Omega\int_{\Dc} \left[ \e(x,y)^{-1} \left( \gradp \ue^0(x) + \gradpy\ue^1(x,y) \right) \cdot
                      \gradpy\bar\ve^1(x,y) \right] \d y\, \d x \,+\\
  + \int_\Omega\int_\D \left[ \e(x,y)^{-1} \left( \gradp\ui^0(x) + \gradpy\ui^1(x,y) \right) \cdot
                      \gradpy\bar\vi^1(x,y) \right] \d y\, \d x \,=\, 0 \\
   \text{for all} \; \ve^1 \in L^2(\Omega;\Honeper(\Dc)), \, \vi^1 \in L^2(\Omega;\Honeper(\D)).
\end{multline}
By using separable test functions $\ve^1(x,y) \mapsto \phi(x)\ve^1(y)$ and $\vi^1(x,y) \mapsto \phi(x)\vi^1(y)$, we obtain the cell problem for each~$x$:
\begin{multline}
  \int_{\Dc} \left[ \e(x,y)^{-1} \left( \gradp \ue^0(x) + \gradpy\ue^1(x,y) \right) \cdot
                      \gradpy\bar\ve^1(y) \right] \d y \,+\\
  + \, \int_\D \left[ \e(x,y)^{-1} \left( \gradp\ui^0(x) + \gradpy\ui^1(x,y) \right) \cdot
                      \gradpy\bar\vi^1(y) \right] \d y \,=\, 0 \\
   \text{for all} \; \ve^1 \in \Honeper(\Dc), \, \vi^1 \in \Honeper(\D).
\end{multline}
Using this together with the definitions (\ref{Pe}--\ref{Pi}) of the corrector matrices $\Pe$ and $\Pi$, we obtain equation \eqref{var2}.
In addition to $v^0=0$, let us also set $\ve^1 = 0$ and $\vi^1(y) = \xxi\dot y$ to obtain
\begin{equation}\label{gradterm}
  \int_\D \e(x,y)^{-1} \left( \gradp\ui^0(x) + \gradpy\ui^0(x,y) \right) \dot \kx\xxi\,\d y\,.
\end{equation}
It follows from this and \eqref{var3} (taken as the definition of $E^0(x,y)$) that
\begin{equation}\label{EoverD}
  \int_\D E^0(x,y) \d A(y) \,=\, 0\,.
\end{equation}
Now let $\vi^0(y)$ be arbitrary and set $\ve^0=0$ and $v^1=0$.  The gradient term vanishes by \eqref{gradterm}, and we obtain
\begin{multline}
 -\omega^2 \int_\Omega \left[ \int_\D \mu(x,y)\d y \right] \ui^0(x)\bar\vi^0(x)\,\d x \,+\\
 +\, \io \int_\Omega \left[ \int_\dD \sigma(x,y)^{-1} \d y \right] \left( \ue^0(x)-\ui^0(x) \right) \bar\vi^0(x)\,\d x = 0,
\end{multline}
or
\begin{equation}\label{jump1}
  \ui^0(x) = \frac{\hat\rho(x)}{\hat\rho(x) - \io\hat\mu(x)} \ue^0(x),
\end{equation}
which, by definition \eqref{two} of $m$, is equation \eqref{var1}.
Now set $\vi^1=0$ and let $\vi^0$ and $\ve^0$ be arbitrary to obtain
\begin{multline}
\int_\Omega \left[ \int_\QQ \e(x,y)^{-1} \left( \grad u^0(x) + \gradpy u^1(x,y) \right) \d y \right]
    \dot \gradp \bar v^0(x) \d x \,+ \\
  -\omega^2 \int_\Omega \left[ \int_\QQ \mu(x,y)M(x,y) \d y \right] \ue^0(x) \bar v^0(x) \d x \,=\, f(v^0).
\end{multline}
This, together with the definition \eqref{Eav2} of $\Eav(x)$ and equation \eqref{EoverD} gives
\begin{equation}
\Eav(x) \,=\, \int_\Dc E^0(x,y) \d A(y),
\end{equation}
and using the definition \eqref{alpha} of $\e^*(x)$ and the definition \eqref{three} of $\mu^*(x)$, we obtain \eqref{var4}.  We observe that equations (\ref{var3}--\ref{var4})
constitute the weak form of the homogenized system \eqref{homogenized}.
One also confirms that equations (\ref{var1}--\ref{var2}) together with the definitions of $\e^*$ and $\mu^*$ imply the two-scale weak-form system \eqref{weaktwoscale}.

Now, we must prove that $\|h^\eta\|_{L^2} < C$ for all $\eta\in\Upsilon'$ so that $u^\eta$ can be replaced by $h^\eta$ in the entire proof.  To do this, we first prove the {\em strong} two-scale convergence of $u^\eta(x)$ to $u^0(x,y)$ in $\Omega_0$, the strong convergence in $\Omega\!\setminus\!\Omega_0$ being standard (see footnote \ref{footnoteStrong} on page \pageref{footnoteStrong}).  By Lemma \ref{extension}, the extensions $\tueeta$ of the restrictions of $u^\eta$ to $\Oeeta$ are bounded in $H^1(\Omega_0)$, and therefore we can restrict $\Upsilon'$ to a subsequence that converges strongly to a function $w(x)$ in $L^2(\Omega_0)$.  By Theorem 1.14 (a) of \cite{Allaire1992}, we can restrict $\Upsilon'$ further so that $\tueeta(x)$ also two-scale converges to a function in $L^2(\Omega_0)$, namely $\ue^0(x)$.  These two facts imply the weak $L^2$ convergence of $\tueeta(x)$ to both $w(x)$ and $\ue^0(x)$, from which we obtain
\begin{equation}
\lim_{\eta\to0} \int_{\Omega_0} (\tueeta(x))^2 \d A(x) = \left[ \int_{\Omega_0} \ue^0(x) \right]^2.
\end{equation}
In a similar manner, we obtain
\begin{equation}
\lim_{\eta\to0} \int_{\Omega_0} (\tuieta(x))^2 \d A(x) = \left[ \int_{\Omega_0} \ui^0(x) \right]^2.
\end{equation}
Now that we have the strong two-scale convergence of $\tueeta(x)$ to $\ue^0(x)$ and $\tuieta(x)$ to $\ui^0(x)$ as well as the two-scale convergence of $\tueeta(x)\chie(x/\eta)$ to $\ue^0(x)\chie(y)$ and $\tuieta(x)\chii(x/\eta)$ to $\ui^0(x)\chii(y)$, we apply Theorem 1.8 of \cite{Allaire1992} (see footnote \ref{footnoteThm18} on page \pageref{footnoteThm18}) to obtain
\begin{multline}\label{strongtsc}
\lim_{\eta\to0} \int_{\Omega_0} \left(u^\eta(x)\right)^2 \d A(x) = 
  \lim_{\eta\to0} \int_{\Omega_0} \left(\tueeta(x)\right)^2 \chie(x/\eta) \,\d A(x) + 
  \lim_{\eta\to0} \int_{\Omega_0} \left(\tuieta(x)\right)^2 \chii(x/\eta) \,\d A(x) \\
  = \int_{\Omega_0} (\ue^0(x))^2 \!\int_\QQ \chie(y) \,\d A(y) \d A(x)
     + \int_{\Omega_0} (\ui^0(x))^2 \!\int_\QQ \chii(y) \,\d A(y) \d A(x) \\ 
      = \int_{\Omega_0} \int_\QQ (\ue^0(x,y))^2 \,\d A (y)\d A(x).
\end{multline}
By definition of $m^\eta$, if $m^0<1$ then $\| u^\eta \|_{L^2(\Omega)} = 1$ for $\eta$ sufficiently small, and by the strong two-scale convergence of $u^\eta$ to $u^0$, we obtain
\begin{equation}
\| u^0 \|_{L^2(\Omega\times\QQ)} = 1 \quad \text{($m_0 < 1$)}.
\end{equation}
Thus $\ue^0\not=0$ because of the relation \eqref{jump1}.  In addition, $\ue^0$ solves the homogenized system \eqref{homogenized} with forcing $m^0f$.   If in fact, if this solution is unique, we conclude that $m^0\not=0$.  We will show momentarily that it is unique at least in the case that $\Im(\mu^*(x))>0$.  We therefore obtain
\begin{equation}
 h^\eta = (m^\eta)^{-1} u^\eta \tsc (m^0)^{-1} u^0 =: h^0,
\end{equation}
obtaining a uniform $L^2$ bound on $h^\eta$.
All of the results in this proof are therefore valid if we replace all occurrences of $u$ with~$h$.
Because of the uniqueness of the solution to the two-scale variational problem, we conclude the convergences obtained in this proof hold for the entire sequence $\Upsilon$.

Finally, we explain why the solution to the homogenized system \eqref{var4} is unique if $\Im(\mu^*(x))>0$.  The argument is standard, so we will be brief.  Setting the right-hand side of the second equation in \eqref{var4} to zero and using $v = \he^0$ as a test function, we obtain
\begin{equation}
\int_\Omega \left[ \e^*(x)^{-1} \gradp\he^0(x)\dot \gradp\bar{\,\he}^{\!\!0}(x) - \omega^2 \mu^*(x) |\he^0(x)|^2 \right] \d A(x) = 0.
\end{equation}
By its definition \eqref{alpha}, $\Im \e^*(x)> 0$ also, and it follows that $\he^0(x)=0$.
\end{proof}
\medskip

\begin{theorem}\label{thm:strong}
The sequence $h^\eta(x)$ converges strongly in the spaces $V^\eta$ to its two-scale limit
$h^0(x,y) + \eta h^1(x,y)$ in the sense that
\begin{equation}
  \lim_{\eta\to0} \| h^\eta(x) - h^0(x,x/\eta) \|_{L^2(\Omega)} = 0 
\end{equation}
and 
\begin{equation}\label{oo}
  \lim_{\eta\to0} \| \grad h^\eta(x) - \gradx h^0(x,x/\eta) - \grady h^1(x,x/\eta) \|_{L^2(\Omega)} = 0.
\end{equation}
\end{theorem}

\begin{proof}
The first limit follows from the strong two-scale convergence of $h^\eta(x)$ to $h^0(x,y)$ (\ref{strongtsc}) and Theorem 1.8 of \cite{Allaire1992}.  To prove the second limit, we write the expansion
\begin{multline}\label{expansion}
b\left(h^\eta(x) - h^0(x,y) - \eta u^1(x,y), h^\eta(x) - h^0(x,y) - \eta h^1(x,y)\right) \\
  = b\left(h^\eta(x), h^\eta(x)\right)
  - b\left(h^\eta(x), h^0(x,y) - \eta h^1(x,y)\right) +\\
  - b\left(h^0(x,y) - \eta h^1(x,y), h^\eta(x) \right) 
  + b\left(h^0(x,y) - \eta u^1(x,y),  h^0(x,y) - \eta h^1(x,y)\right).
\end{multline}
The first term on the right-hand side is equal to $\omega^2 c\left(h^\eta(x),h^\eta(x)\right) + f\left(h^\eta(x)\right)$,
which, by the strong two-scale convergence of $h^\eta(x)$ to $h^0(x,y)$ tends to
\begin{equation}\label{ooo}
\omega^2 \int_\Omega \mu^*(x) \left( \he^0(x) \right)^2 \d A(x) + f\left(\he^0(x)\right).
\end{equation}
By using the admissibility of $h^0 + \eta h^1$ as a test function for two-scale convergence together with Theorem 1.8 of \cite{Allaire1992} and its analog for two-scale convergence on hypersurfaces, we can pass to the two-scale limit in each of the other three terms, which is (plus or minus depending on the sign in \eqref{expansion})
\begin{multline}
   \int_\Omega \int_\QQ \e(x,y)^{-1} \left[ \gradpx h^0(x,y) + \gradpy h^1(x,y) \right] \cdot
           \left[ \gradpx \overline{h^0}(x,y) + \gradpy \overline{h^1}(x,y) \right] \,\d A(y) \d A(x) +\\
   - \io \int_\Omega \left[ \int_\dD \sigma(x,y)^{-1} \d s(y) \right] \left| \he^0(x)-\hi^0(x) \right|^2 \d A(x).
\end{multline}
By Theorem \ref{thm:vartwoscale} and the definition of $\mu^*$, this is equal to \eqref{ooo}, and we find that the right-hand side of \eqref{expansion} tends to zero with $\eta$.  By the coercivity of $b$ (\ref{coercivity}), the limit \eqref{oo} now follows.
\end{proof}

\begin{theorem}\label{thm:transmission}
The reflection and transmission coefficients for the problems of scattering by the micro-structured slabs (Problem \ref{microweak}) converge to those of the problem of scattering by the homogenized slab (equation \ref{var4}).
\end{theorem}

\begin{proof}
This follows from the convergence of the Fourier coefficients of the propagating harmonics of the solutions $h^\eta$ to those of the solution $\he^0$,
\begin{equation}
\lim_{\eta\to0} \int_{\Gamma_\pm} h^\eta(x) e^{-\ii(m+\kappa)x_2} \d x_2
 = \int_{\Gamma_\pm} \he^0(x) e^{-\ii(m+\kappa)x_2} \d x_2,
\end{equation}
which is valid because of the strong convergence of the solutions and their gradients outside the slab.
\end{proof}

\bibliography{KohnShipman}

\begin{thebibliography}{10}

\bibitem{Allaire1992}
Gregoire Allaire.
\newblock Homogenization and two-scale convergence.
\newblock {\em SIAM Journal on Mathematical Analysis}, 23(6):1482--1518, 1992.

\bibitem{AllaireDamlamianHornung1995}
Gr{\'e}goire Allaire, Alain Damlamian, and Ulrich Hornung.
\newblock {\em Two-scale convergence on periodic surfaces and applications},
  pages 15--25.
\newblock Mathematical Modelling of Flow through Porous Media. World
  Scientific, Singapore, 1995.

\bibitem{Bonnet-BeStarling1994}
Anne-Sophie Bonnet-Bendhia and Felipe Starling.
\newblock Guided waves by electromagnetic gratings and nonuniqueness examples
  for the diffraction problem.
\newblock {\em Math. Methods Appl. Sci.}, 17(5):305--338, 1994.

\bibitem{BouchitteFelbacq2004}
Guy Bouchitt{\'e} and Didier Felbacq.
\newblock Homogenization near resonances and artificial magnetism from
  dielectrics.
\newblock {\em C. R. Acad. Sci. Paris}, I(339):377--382, 2004.

\bibitem{BouchitteFelbacq2005}
Guy Bouchitt{\'e} and Didier Felbacq.
\newblock Theory of mesoscopic magnetism in photonic crystals.
\newblock {\em Phys. Rev. Lett.}, 2005.

\bibitem{BouchitteFelbacq2006}
Guy Bouchitt{\'e} and Didier Felbacq.
\newblock Homogenization of a wire photonic crystal: The case of small volume
  fraction.
\newblock {\em SIAM J. Appl. Math.}, 2006.

\bibitem{CherednicSmyshlyaeZhikov2006}
K.~D. Cherednichenko, V.~P. Smyshlyaev, and V.~V. Zhikov.
\newblock Non-local homogenized limits for composite media with highly
  anisotropic fibres.
\newblock {\em Proc. R. Soc. Edinburgh}, 136A:87--114, 2006.

\bibitem{CioranescSaint-Jea1979}
Doina Cioranescu and Jeannine Saint Jean~Paulin.
\newblock Homogenization in open sets with holes.
\newblock {\em J. Math. Anal. Appl.}, 71(2):590--607, October 1979.

\bibitem{FelbacqBouchitte2005a}
Didier Felbacq and Guy Bouchitt{\'e}.
\newblock Left-handed media and homogenization of photonic crystals.
\newblock {\em Optics Letters}, 30(10):1189--1191, 2005.

\bibitem{FelbacqBouchitte2005}
Didier Felbacq and Guy Bouchitt{\'e}.
\newblock Negative refraction in periodic and random photonic crystals.
\newblock {\em New Journal of Physics}, 159, 2005.

\bibitem{GuenneauZolla2007}
S.~Guenneau and F.~Zolla.
\newblock Homogenization of 3d finite chiral photonic crystals.
\newblock {\em Phys. D}, 394:145--147, 2007.

\bibitem{Neuss-Rad1996}
Maria Neuss-Radu.
\newblock Some extensions of two-scale convergence.
\newblock {\em CRAS, Paris, Series I, Math.}, 322(9):899--904, 1996.

\bibitem{OBrienPendry2002}
S.~O'Brien and J.~B. Pendry.
\newblock Magnetic activity at infrared frequencies in structured metallic
  photonic crystal.
\newblock {\em J. Phys.: Condens. Matter}, 14:6383---6394, 2002.

\bibitem{OBrienPendry2002a}
Stephen O'Brien and John~B. Pendry.
\newblock Photonic band-gap effects and magnetic activity in dielectric
  composites.
\newblock {\em J. Phys.: Condens. Matter}, 14:4035--4044, 2002.

\bibitem{Pendry2004a}
J.~B. Pendry.
\newblock A chiral route to negative refraction.
\newblock {\em Science}, 306(5700):1353--1355, 2004.

\bibitem{PendryHoldenRobbins1998}
J.~B. Pendry, A.~J. Holden, D.~J. Robbins, and W.~J. Stewart.
\newblock Low frequency plasmons in thin-wire structures.
\newblock {\em J. Phys.: Condens. Matter}, 10:4785--4809, 1998.

\bibitem{PendryHoldenRobbins1999}
J.~B. Pendry, A.~J. Holden, D.~J. Robbins, and W.~J. Stewart.
\newblock Magnetism from conductors and enhanced nonlinear phenomena.
\newblock {\em IEEE Trans. Microw. Theory Tech.}, 47(11):2075--2084, 1999.

\bibitem{Pendry2004}
John~B. Pendry.
\newblock Negative refraction.
\newblock {\em Contemporary Physics}, 45(3):191--202, May--June 2004.

\bibitem{Ramakrish2005}
S.~Anantha Ramakrishna.
\newblock Physics of negative refractive index materials.
\newblock {\em Rep. Prog. Phys.}, 68:449--521, 2005.

\bibitem{SievenpipYablonoviWinn1998}
D.~F. Sievenpiper, E.~Yablonovitch, J.~N. Winn, S.~Fan, P.~R. Villeneuve, and
  J.~D. Joannopoulos.
\newblock 3d metallo-dielectric photonic crystals with strong capacitive
  coupling between metallic islands.
\newblock {\em Phys. Rev. Lett.}, 80(13):2829--2832, Mar 1998.

\bibitem{SmithPadillaVier2000}
D.~R. Smith, Willie~J. Padilla, D.~C. Vier, S.~C. Nemat-Nasser, and S.~Schultz.
\newblock Composite medium with simultaneously negative permeability and
  permittivity.
\newblock {\em Phys. Rev. Lett.}, 84(18):4184--4187, May 2000.

\bibitem{SmithSchultzMarko2002}
D.~R. Smith, S.~Schultz, P.~Marko, and C.~M. Soukoulis.
\newblock Determination of effective permittivity and permeability of
  metamaterials from reflection and transmission coefficients.
\newblock {\em Phys. Rev. B}, 65(19):195104, Apr 2002.

\bibitem{SoukoulisKafesakiEconomou2006}
Costas~M. Soukoulis, Maria Kafesaki, and Eleftherios~N. Economou.
\newblock Negative-index materials: New frontiers in optics.
\newblock {\em Advanced materials}, 18(15):1941--1952, 2006.

\end{thebibliography}

\end{document}